\documentclass[11pt,letterpaper]{amsart}

\usepackage[utf8]{inputenc}
\usepackage[T1]{fontenc}

\usepackage{times}

\usepackage{lmodern}

\usepackage{dutchcal}
\usepackage{tikz}
\usetikzlibrary{decorations.pathmorphing,patterns}
\usetikzlibrary{calc,arrows}
 
\usepackage{dutchcal}

\usepackage{caption}

\usepackage[T1]{fontenc}

\usepackage{media9}
\usepackage{animate}
\usepackage{graphicx}

\usepackage{amssymb,fge}

\usepackage{scrextend}
\usepackage{lipsum}

\usepackage{ulem}

\usepackage[pagebackref,colorlinks=true,pdfpagemode=none,urlcolor=blue,
linkcolor=blue,citecolor=blue]{hyperref}

\usepackage{amsmath,amsfonts,amssymb,amsthm}
\usepackage{amsthm, amssymb, amsmath, amsfonts, mathrsfs}

\usepackage{mathrsfs}
\usepackage{MnSymbol}
\usepackage{scalerel} % for scaling the cube

\usepackage{tikz}
\usepackage{color}
\usepackage{accents}

\usepackage{bbm}

\definecolor{labelkey}{gray}{.8}
\definecolor{refkey}{gray}{.8}

\definecolor{darkred}{rgb}{0.9,0.1,0.1}
\definecolor{darkgreen}{rgb}{0,0.5,0}

\newcommand{\tnorm}[1]{{\left\vert\kern-0.25ex\left\vert\kern-0.25ex\left\vert #1 
    \right\vert\kern-0.25ex\right\vert\kern-0.25ex\right\vert}}

\newtheorem{theorem}{Theorem}[section]

\newtheorem{definition}[theorem]{Definition}
\newtheorem{corollary}[theorem]{Corollary}
\newtheorem{proposition}[theorem]{Proposition}

\newtheorem{example}{Example}
\newtheorem{assumption}{Assumption} 
 
\theoremstyle{remark}
\newtheorem{remark}[theorem]{Remark}

\renewenvironment{proof}[1][Proof]{ {\itshape \noindent {#1.}} }{$\Box$
\medskip}

\numberwithin{equation}{section}
\newcommand{\commentout}[1]{}
\newcommand{\bes}{\begin{displaymath}}
\newcommand{\ees}{\end{displaymath}}
\newcommand{\be}{\begin{equation}}
\newcommand{\ee}{\end{equation}}
\newcommand{\ba}{\begin{eqnarray}}
\newcommand{\ea}{\end{eqnarray}}
\newcommand{\bas}{\begin{eqnarray*}}
\newcommand{\eas}{\end{eqnarray*}}

\newcommand{\dd}{{\rm d}}

\newcommand{\cal}{\mathcal}

\newcommand{\rv}{{\bf r}}
\newcommand{\pv}{{\bf p}}

\newcommand{\bbE}{\Bbb E}
\newcommand{\bbR}{\Bbb R}

\newcommand{\bbZ}{{\mathbb Z}}

\newcommand{\ga}{\gamma}
\newcommand{\al}{\alpha}
\newcommand{\bbP}{\Bbb P}

\newcommand{\la}{\lambda}

\newcommand{\Om}{\Omega}

\newcommand{\mc}[1]{{\mathcal #1}}

\newcommand\lang{\langle\langle}
\newcommand\rang{\rangle\rangle}

\definecolor{lviolet}{rgb}{0.5,0,0.5}
\definecolor{ao}{rgb}{0.0, 0.5, 0.0}
\definecolor{byl}{rgb}{0.0,1.0,0.0}

\begin{document}

\title[Superdiffusive energy transport]{Boundary Thermalization in Superdiffusive Energy Transport}

% \title[$3$ conservation laws]{{Superdiffusive energy evolution in
%  a one-dimensional open harmonic chain connected to two heat bath at the boundaries.}} 

\author{Tomasz Komorowski}
\address{Tomasz Komorowski\\ IMPAN,  \'{S}niadeckich 8, 00-656, Warsaw, Poland
  \\
  ORCID \# 0000-0003-1564-0169}

\email{\tt tkomorowski@impan.pl}

 \author{Stefano Olla}
 \address{Stefano Olla, CEREMADE,
   Universit\'e Paris-Dauphine, PSL Research University \\
 \and Institut Universitaire de France\\ \and
 GSSI, L'Aquila\\
ORCID \# 0000-0003-0845-1861}
  \email{\tt olla@ceremade.dauphine.fr}

% --- Dedication ---
\dedicatory{Dedicated to Claudio Landim on the occasion of his 60th birthday,\\ for our long friendship and collaboration.}

\begin{center}
    \vspace{-2em} % Optional: Adjusts vertical spacing if needed
\end{center}

  \begin{abstract}
We study energy transport in a finite one-dimensional unpinned harmonic chain with stochastic nearest-neighbor momentum exchanges and Langevin heat baths at its endpoints. Such systems are known to exhibit superdiffusive transport driven by long-wavelength acoustic modes, leading to fractional macroscopic behavior. While fractional heat equations have been rigorously derived for infinite chains, the corresponding boundary conditions for finite systems in contact with heat baths remain unclear due to the nonlocality of the fractional Laplacian.
Under the superdiffusive {time} scaling  $t\sim n^{3/2}$,
{where $n$ is the system size,}
we prove that the averaged microscopic energy profile converges,
{as $n\to+\infty$,} to a temperature field solving a fractional heat equation on $[0,1]$,
 with the generator given by a Neumann fractional Laplacian and
 additional nonlocal boundary terms induced by the {heat} baths.
Our results provide a
 rigorous derivation of macroscopic boundary conditions for
 superdiffusive heat transport in open chains and introduce new
 boundary conditions for fractional Laplacians, {that are motivated by
 a physical model.}
  \end{abstract}

  \thanks{T.K acknowledges the support of
the NCN grant 2024/53/B/ST1/00286}
% This paper has been partially supported by }
% \keywords{...}
% \subjclass[2000]{...}
\maketitle

\section{Introduction}

Heat conduction in dielectric solids is mediated by lattice vibrations that propagate as waves.
In one-dimensional unpinned (acoustic) chains, long-wavelength modes travel with
non-vanishing velocity and are expected to
produce superdiffusive transport of energy.
The thermal conductivity can be expressed through
the Green–Kubo formula,
involving space–time correlations of the microscopic dynamics,
and numerical studies of nonlinear oscillator chains
suggest that it diverges
with system size \cite{sll,llp97}.
However, rigorous results for deterministic anharmonic
systems remain largely out of reach,
and even the convergence of the Green–Kubo formula
{(when expected finite)} is not proven, see \cite{BLR}.
Mathematical progress has instead been achieved for harmonic chains
perturbed by stochastic exchanges of momenta between neighboring particles
\cite{bborev,bbo2,bo11,BBJKO,JKO15,BGJ16}.
This noise conserves energy and momentum while breaking integrability,
induces scattering with rates inversely proportional to {the square of the wavelength},
and leads macroscopically to a L\'evy-type superdiffusion described by a
fractional heat equation with the generator proportional to the fractional Laplacian
$|\Delta|^{3/4}$.

Most rigorous results to date concern infinite systems without
boundaries. When a finite chain is coupled to heat baths at its
endpoints, a natural question is which macroscopic boundary conditions
emerge for the resulting superdiffusive energy evolution. In diffusive
settings (e.g. pinned chains or noise that breaks momentum
conservation), the baths impose Dirichlet boundary conditions fixed by
their temperatures. For a fractional diffusion, however, the nonlocality of the fractional Laplacian allows several possible boundary conditions, and it is not a priori clear which ones arise from the microscopic dynamics. {Physics literature} \cite{LM08,LM10,KB18} predicts that the macroscopic boundary behavior depends on whether the chain is pinned or unpinned at the endpoints, but a rigorous derivation has been missing as well as
a
% \sout{non-rigorous precise}}
prediction of the boundary layer arising in the unpinned case.

Some results for {infinite systems with a boundary consisting of a single
  Langevin thermal bath have been obtained} in the 'kinetic limit'.
{This correspond to a  {\it two step approach}, where first 
  the random exchange is rescaled so that in the macroscopic unit
time only a finite number of stochastic exchanges occurs
\cite{BOS}.
The corresponding kinetic equation is then rescaled superdiffusively
that, in the limit,  leads to
a  fractional evolution of the energy  with a
boundary condition at one site \cite{kjo,babo}.
In the presence of a single heat bath this program was
  achieved in \cite{KORS20,KO20} for the kinetic
  equation with a boundary term, and  \cite{KOR20}
  for the corresponding
  hydrodynamic limit (see \cite{KOrev} for a complete review
  of these results).}

In this paper we review our recent result \cite{KO-supp}
concerning a finite unpinned, {one-dimensional}  harmonic chain
with stochastic nearest-neighbor momentum exchanges and
Langevin heat baths
at its boundaries. The dynamics conserves energy, momentum, and total
length {$n$} in the bulk,
while the baths exchange energy at fixed temperatures.
Starting from suitable initial data and adopting the superdiffusive scaling $t\sim n^{3/2}$,
we analyze the evolution of the averaged microscopic energy profiles.
As the system size {$n$} tends to infinity, these profiles converge to a
macroscopic temperature field solving a fractional heat equation on $[0,1]$
with the generator given by a Neumann fractional Laplacian
with additional
nonlocal boundary terms {produced}  by the heat baths.
These terms can be interpreted in terms of
absorption and creation rates for a
Lévy-type process and enforce the bath temperatures at the endpoints,
while, {at the same time, produce} boundary layers characteristic of unpinned chains.

{Our results provide a rigorous derivation of the macroscopic boundary conditions governing superdiffusive heat transport in open chains.}
 %\sout{thereby confirming previously obtained non-rigorous predictions.}}
In addition,  the emerging boundary conditions for fractional
Laplacians  appear to be new
in the mathematical literature.
They are expected to be
relevant for more general nonlinear chains, where numerical simulations
are typically performed in contact with heat baths and exhibit similar superdiffusive behavior.

{In the framework of mass transport, different boundary conditions for
the regional fractional Laplacian were obtained in the hydrodynamic limit
for exclusion processes with long jumps in contact with spacially
extended particle reservoirs \cite{BGJO-21,BCGS-24}. }

\subsection{Notations}

We denote \begin{equation}
  \label{nbl}
  \begin{split}
    &\nabla  f_x = f_{x+1} - f_{x} \quad \mbox{and}\\
    &\nabla^\star f_x =    f_x - f_{x-1},\quad x \in \{0,1,\ldots,n\} =: \bbZ_n.
    \end{split}
\end{equation}
 We adopt the convention 
 $
    f_{n+1}:=f_n$ and $ f_{-1}= f_0.
    $
The Neumann discrete Laplacian, $\Delta_{\rm N}$ is defined as
\begin{align*}
  \Delta_{\rm N}f_x= f_{x+1}+f_{x-1}-2f_x = \nabla \nabla^\star f_x, \quad x\in \bbZ_n.
    \end{align*}
    In the following we denote by $\Delta$ the usual (continuous)
    Laplacian on $[0,1]$ with the Neumann boundary
    conditions.

\subsection{Model: stochastically perturbed unpinned harmonic chains}

  We consider   a finite chain of $n+1$ atoms, with
two Langevin  heat baths, at temperatures $T_L$ 
and $T_R$,  attached to the left and right endpoints, respectively,
see Figure \ref{fig1} below.
 The dynamics   of the positions
$\mathbf q(t)=\big(q_x(t)\big)_{x\in\bbZ_n}$ and momenta
$\mathbf p(t)=\big(p_x(t)\big)_{x\in\bbZ_n}$,
where $q_x(t),p_x(t)\in\bbR$ are given by:
\begin{equation} \label{eq:bas1}
  \begin{split}
    \dot q_x(t) =&\ p_x(t), % \qquad x\in\bbZ_n:=\{0,\ldots,n\},
    \\
   \dd p_x(t) =&\ \Delta_{\rm N} q_x(t)\dd t +
   \big[\nabla^\star p_{x+1}(t-)\;\dd N_{x,x+1}(\gamma t)
   \\
   &
 - \nabla^\star p_x(t-)\; \dd N_{x-1,x}(\gamma t)\big]  +\sum_{z=0,n}\delta_{x,z}\big(- \tilde\gamma  p_{z}\dd t + \sqrt{2T_z \tilde \gamma} \dd
     w_{z}\big).
\end{split}
\end{equation}
{Here $w_0(t)$ and $w_n(t)$ are independent standard Brownian motions, and 
$\{N_{x,x+1}(t), x=0,\dots, n-1\}$ are independent Poisson processes of intensity $1$,
independent of the Brownian motions.
These processes are defined over
a probability space $(\Om,{\cal F},\bbP)$.
The parameters $\ga,\tilde
\ga>0$ determine 
the respective rates of the momentum exchange and the strength
of the heat baths.
To simplify the notation, at the boundaries $x=0,n$ we shall also write sometimes $T_0=T_L$,
$w_0=w_L$ and $T_n=T_R$,
$w_n=w_R$.

Being motivated by the results of \cite{JKO15} we consider the dynamics under the 
superdiffusive time scaling  $t':=t/n^{3/2}$, $u=x/n$,
where  {$(t',u)\in [0,+\infty)\times[0,1]$} correspond to the macroscopic time and space co-ordinates.
Since  the dynamics is unpinned, it  is  invariant under
translations: $q_x \rightarrow q_x + a$, $a\in\bbR$. For this reason,
it is   convenient to work with    the inter-particle
  stretches
$
r_x:=q_x-q_{x-1}$, 
$x=1,\ldots,n$, $r_0:=0.
$
Then, the scaled dynamics  of
  $$
r^{(n)}_x(t)=r_x(n^{3/2} t),\quad p^{(n)}_x(t)=p_x(n^{3/2} t)
$$
is described by the following system of equations
\begin{align}
  \label{011012-25}
    \dot   r_x^{(n)}(t) = & n^{3/2}\nabla^\star p_x^{(n)}(t),\notag\\
    \dd   p_x^{(n)}(t) = & n^{3/2}\big(\nabla r_x^{(n)} -\ga
                           \Delta_{\rm N}p_x^{(n)}(t) \big)\dd t \\
                          &
                            +
                          { \nabla\left[% \nabla^\star  p_{x+1}^{(n)}(t-)\dd
      \nabla^\star
      p_x^{(n)}(t-)\dd \tilde N_{x-1,x}(n^{3/2}\gamma t)\right]}\notag\\
      &+\sum_{z=0,n}\delta_{x,z}\big[-n^{3/2}
                       \tilde\gamma p_z^{(n)} (t) \dd t
                     +\sqrt{2 n^{3/2}\tilde  \gamma T_z} \dd   w_z(t)\big]. \notag
\end{align}
Here  $\tilde N^{(n)}_{x,x+1}(t)=N^{(n)}_{x,x+1}(t)-t$.

  To lighten somewhat the  notation we shall omit writing the
  superscript $n$ by the stretches and momenta, using from this point
  on  the  convention:
  $ r_x (t) :=r_x^{(n)} (t)$ and   $p_x (t):=  p_x^{(n)} (t)$.

 The generator of the dynamics is given by   $n^{3/2}\mathcal G $, with
\begin{equation}
  \label{eq:7}
  \mathcal G =  \mathcal A +  \gamma S_{\text{ex}}
  +    \tilde\gamma (S_L+  S_R),
\end{equation}
where, with the convention $r_0=r_{n+1}=0$,
\begin{itemize}
\item its Hamiltonian part equals
\begin{equation}
  \label{eq:8}
  \mathcal A= \sum_{x=1}^n \nabla^\star p_x \partial_{r_x}
  + \sum_{x=0}^n  \nabla r_{x} \partial_{p_x},
\end{equation}
\item the momentum exchange part is
\begin{equation}
  \label{eq:21}
  \begin{split}
 &   S_{\text{ex}} f (\rv,\pv) =    \sum_{x=0}^{n-1}   \Big( f
 (\rv,\pv^{x,x+1}) - f (\rv,\pv)\Big).
\end{split}
  \end{equation}
 for any bounded and measurable function $f:\bbR^n\times\bbR^{n+1}\to\bbR$. 

  \item   the action of 
 thermostats corresponds to 
 \begin{equation}
   \label{eq:10}
   S_{L} = T_L \partial_{p_0}^2 - p_0 \partial_{p_0}\quad  S_{R} = T_R \partial_{p_n}^2 - p_n \partial_{p_n}.
 \end{equation}

 \end{itemize}
Here
 $\pv^{x,x'}$ denotes the momentum configuration, where the velocities at
 sites {$x\not=x'$} have been exchanged, 
   i.e. $\pv^{x,x'}=(p^{x,x'}_0,\ldots,p_n^{x,x'})$, with $p_y^{x,x'}=p_y$,
   $y\not \in\{x,x'\}$ and  $p_{x'}^{x,x'}=p_x$, $p_{x}^{x,x'}=p_{x'}$.

   \begin{figure}
\begin{center}
  \begin{tikzpicture}
   
\node[circle,fill=black,inner sep=1.2mm] (e) at (0,0) {};
\node[circle,fill=black,inner sep=1.2mm] (f) at (2,0) {};
\node[circle,fill=black,inner sep=1.2mm] (g) at (3.5,0) {};
\node[circle,fill=black,inner sep=1.2mm] (h) at (4.8,0) {};
\node[circle,fill=black,inner sep=1.2mm] (i) at (6.8,0) {};
\node[circle,fill=black,inner sep=1.2mm] (j) at (8.5,0) {};
\node[circle,fill=black,inner sep=1.2mm] (k) at (10,0) {};
\node[circle,fill=black,inner sep=1.2mm] (l) at (11.3,0) {};

\draw[dashed] (2,0) -- (3.5,0);
\draw[dashed] (8.5,0) -- (10,0);
%\draw[ultra thick, blue, ->] (11.3,0) -- (12.5,0);

\draw[thick, <->] (3.5,-1.2) -- (4.8,-1.2);

\draw (0, -0.6) node[] {$q_{0}$};
\draw (3.5, -0.6) node[] {$q_{x-1}$};
\draw (4.8, -0.6) node[] {$q_{x}$};
\draw (6.8, -0.6) node[] {$q_{x+1}$};
\draw[dashed] (3.5,-1.5) -- (3.5,-1);
\draw (4.1, -1.5) node[] {$r_x$};
\draw[dashed] (4.8,-1.5) -- (4.8,-1);

\draw (-0.6,1.8) node[] {\large\color{blue}$T_-$};
\draw (12,1.8) node[] {\large\color{red}$T_+$};
\draw (11.3,-0.6) node[] {$q_{n}$};;

\draw[decoration={aspect=0.3, segment length=3mm, amplitude=3mm,coil},decorate] (0,0) -- (2,0); 
\draw[decoration={aspect=0.3, segment length=1.8mm, amplitude=3mm,coil},decorate] (3.5,0) -- (4.9,0); 
\draw[decoration={aspect=0.3, segment length=3mm, amplitude=3mm,coil},decorate] (4.8,0) -- (6.9,0); 
\draw[decoration={aspect=0.3, segment length=2.5mm, amplitude=3mm,coil},decorate] (6.8,0) -- (8.6,0); 
\draw[decoration={aspect=0.3, segment length=1.8mm, amplitude=3mm,coil},decorate] (10,0) -- (11.4,0); 

\fill [pattern = north east lines, pattern color=blue] (-0.3,0.8) rectangle (0.3,2);
\fill [pattern = north east lines, pattern color=red] (11,0.8) rectangle (11.6,2);
\node (c) at (-0.3,1.5) {};
\node (d) at (-0.1,0.1) {};
\node (a) at (11.6,1.5) {};
\node (b) at (11.4,0.1) {};

%\draw[thick, ->, bend left=45]  (c) -- (d);
%\draw[thick] (-1,5) -- (3,5);
\draw (c) edge[dashed, ultra thick, blue, ->, >=latex, bend right=60] (d);
\draw (a) edge[dashed, ultra thick, red, ->, >=latex, bend left=60] (b);

\end{tikzpicture}
\end{center}
\captionof{figure}{\small Oscillator chain: thermostats at both endpoints}
  \label{fig1}
\end{figure}
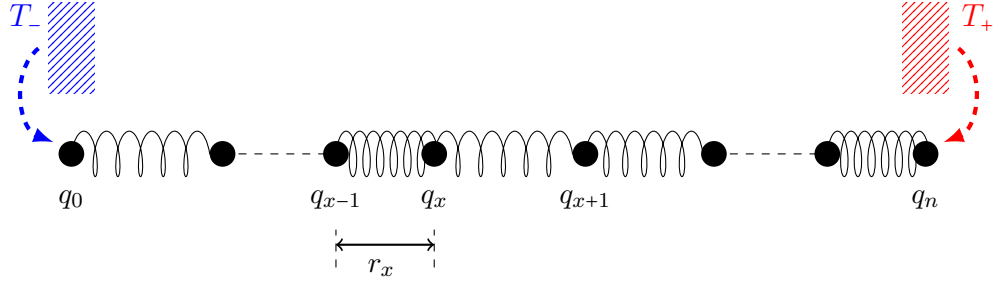

\subsubsection{Conserved quantities and calculation of energy currents}

 {Let $ {\cal E}_x(t)=  {\cal E}_x
  \big(\mathbf r(t), \mathbf p(t)\big)$, where
      $$
      {\cal E}_x (\mathbf r, \mathbf p):=  \frac{1}2 (p_x^2+
r_{x}^2),\
\quad x\in\bbZ_n,
$$
be   the energy per atom (or microscopic energy density).}
Observe that for  $x\in\bbZ_n$
 \begin{align}
   \label{current}
&
  \frac {d}{dt}\bbE_n\left[ {\cal E}_{x}(t)\right]
    =n^{3/2}\mathcal G\mathcal E_{x}(t)   =-n^{3/2}\nabla^\star j_{x,x+1}  (t), \quad \mbox{where}\notag
\\
     &
       j_{x,x+1}(t)= j_{x,x+1}^{(a)}(t)+j_{x,x+1}^{(s)}(t)\\
  & j_{x,x+1}^{(a)}(t):=- p_{x}(t) r_{x+1}(t) , \quad
    j_{x,x+1}^{(s)}=-\frac{\ga}{2}\big(p_{x+1}^2-p_{x}^2\big)
, \notag\\
   &  { j_{-1,0}:= \tilde{ \gamma} \left(T_L - p_{0}^2 \right)},
      \qquad
  j_{n,n+1} :=      \tilde{ \gamma} \left( p_{n}^2-T_R \right). \notag
\end{align}
The quantities $j_{x,x+1}(t)$ are called the {\it energy currents}. We
shall refer to $j_{x,x+1}^{(a)}(t)$, $j_{x,x+1}^{(s)}(t)$ as the
{\it mechanical} and {\it stochastic} parts of the current,
respectively.

{As can be seen from \eqref{011012-25} and \eqref{current}, in the
  case $\tilde\gamma=0$, i.e. when the system is closed, the dynamics   conserves 
 the total stretch  $  \sum_{x\in\bbZ_n}r_x (t)$, 
  momentum $\sum_{x\in\bbZ_n}p_x(t)$ and  energy
  $\sum_{x\in\bbZ_n}{\cal E}_x(t)$.
 }

\subsection{Assumptions about the initial data}

\label{sec1.4}

\subsubsection{The existence of the initial profile of energy}

 We suppose that  the  distribution $\mu_n$ of the initial data $({\bf
  r}(0),{\bf p}(0))\in\Om_n:=\bbR^n\times \bbR^{n+1}$  is   of   zero mean, i.e. 
$$
\int_{\Om_n}r_x\mu_n(\dd {\bf r},\dd{\bf p})=0\quad \mbox{and}\quad
\int_{\Om_n}p_x\mu_n(\dd {\bf r},\dd{\bf p})=0 ,\quad \mbox{for }x\in\bbZ_n.
$$
Furthermore, we assume   the following.
\begin{assumption}
\label{IP} There exists   an initial profile of energy
$T_{\rm ini} \in L^2[0,1]$ such
that   for any $\varphi\in C[0,1]$
  \begin{align*}
   \lim_{n\to+\infty}\frac{1}{n}\sum_{x=0}^{n}  \varphi\left(\frac x{n} \right)
  \int_{\Om_n} {\cal E}_{n,x} \dd \mu_n=\int_0^1 T_{\rm ini} (u)\varphi(u)\dd
  u.
  \end{align*}  
\end{assumption}
By $\bbE_n$ we shall denote the expectation w.r.t. the measure
$\mu_n\otimes \bbP$.

  \subsubsection{Entropy bound}
For a given $T>0$, define the Gaussian probability measure
\begin{equation*} 
  \begin{split}
    &\nu_{T} (\dd{\bf r},\dd{\bf p}) : =g_{T}({\bf r},{\bf p})  \dd{\bf r}\dd{\bf p} ,
    \quad \mbox{where }
    \quad 
    g_{T}({\bf r},{\bf p})=\frac{e^{-\mc E_0/T}}{\sqrt{2\pi
        T}}\prod_{x=1}^n \frac{e^{-\mc E_x/T}}{2\pi T} .
    \end{split}
  \end{equation*}
Assume that $\mu_n$ is absolutely continuous w.r.t. $\nu_{T} $, with
the density $f_n$. The relative entropy of $\mu_n$ w.r.t. $\nu_{T} $ is defined as
\begin{equation}
  \label{rel-ent}
  {\mathbf{H}}_{n,T} =\int_{\Om_n}  f_n \log f_n \dd
{\nu_{T}}=\int_{\Om_n}   \log f_n \dd
\mu_n.
\end{equation}
It is an easy consequence of the Jensen inequality that  ${\mathbf{H}}_{n,T} \ge0$.

  \begin{assumption}
    \label{ass1}
    We suppose that  the relative entropy of the initial distribution
    with respect to $\nu_{T}$ is bounded by
 the size of the system, i.e. there exists a constant {$C_{T} >0$}
 such that
 \begin{equation}
   \label{011012-25z}
 {\mathbf{H}}_{n,T}  \le C_{T} n,\quad n=1,2,\ldots.
\end{equation}
\end{assumption}

\begin{example}[Local equilibrium]
  \label{ex011012-25}
  Suppose that $\beta:[0,1]\to(0,+\infty)$. Let
  $\beta_x := \beta(x/(n+1))$, $x=0,\ldots,n+1$. 
Let $\nu_\beta$ be the probability  measure on $\Om_n$ given by the
formula 
\begin{equation}
  \label{tilde-nu}
  \begin{split}
    &\nu_\beta (\dd \mathbf{r},\dd \mathbf{p}) :=
    { \frac{e^{-\beta_0 \mc E_0/T}}{\sqrt{2\pi \beta_0^{-1}}}
    \prod_{x=1}^n \frac{e^{-\beta_x \mathcal E_x}}{2\pi \beta_x^{-1}} }
    \dd \mathbf{r}\dd \mathbf{p}.
   \end{split}
\end{equation}
% Here 
% \begin{equation}
%   \label{eq:16g}
%   \mathcal G(\beta):= \log\int_{\mathbb R^2} e^{-\frac{\beta}{2} (r^2+p^2)} \dd p \dd r =
%   \log \left(2\pi \beta^{-1}\right),\quad \beta>0,
% \end{equation}
% is  called  the Gibbs potential.
Suppose that  $\mu_n=\nu_\beta$.
The density    with respect to $\nu_T$ satisfies
\begin{equation}
  f_n (\mathbf{r},\mathbf{p})=
  \prod_{x=0}^n \exp\Big\{\Big(\frac{1}{T}-\beta_x \Big)\mathcal E_x
  { - \left(1- \frac{\delta_{x,0}}{2}\right) \log(T\beta_x)} \Big\}   . \label{eq:tildef}
\end{equation}
The relative entropy is then given by 
\begin{equation}
  \label{021012-25}
  \mathbf{H}_{n,T}  
   =\sum_{x=0}^n \int_{\Om_n}\Big(\frac{1}{T}-\beta_x \Big) \mathcal E_x\dd\mu_n
    - \sum_{x=0}^n  \left(1- \frac{\delta_{x,0}}{2}\right) \log(T\beta_x).
\end{equation}
Suppose now that $\beta(\cdot)$ is a bounded function.
Then, under the Assumption
\ref{IP}, we conclude that, the  terms on the
right hand side of \eqref{021012-25} can be estimated by
$$
n(1+o_n(1))\Big(\frac{1}{T}+\|\beta\|_\infty  \Big) \|T_{\rm
  ini}\|_{L^1[0,1]},
$$
and
$n\|\log\big(T\beta(\cdot)\big)\|_\infty$, respectively.
Here, $\lim_{n\to+\infty}o_n(1)=0$
and, for a given function $f:[0,1]\to\bbR$ we let 
$\|f\|_\infty :=\sup_{u\in[0,1]}|f(u)|$.  Thus,  
Assumption \ref{ass1} holds.
   \end{example}

\subsubsection{Square summability of covariances}
    
  Define
  \begin{align}
    \label{Hn2a}
    {\cal H}_n^{(2)}(t) 
=\frac 1{2n} 
                           \sum_{x,x'=0}^n &\bigg\{\Big[\mathbb E_{n}\left( p_{x}(t) p_{x'}(t)\right)\Big]^2+
                           \Big[\mathbb E_{n}\left( r_x(t) r_{x'}(t)\right)\Big]^2 
  \\
  &
    +2 \Big[\mathbb E_{n}\left( p_x(t) r_{x'}(t)\right)\Big]^2\bigg\}. \notag
\end{align}
We assume the following.
\begin{assumption}[Bound on the covariances of the initial data]
  \label{ass2}
  There exists $C_{2,{\cal H}}>0$ such that
  \begin{equation}
    \label{Hn2}
        {\cal H}_n^{(2)}(0)  \le C_{2,{\cal H}},\quad n=1,2,\ldots.\end{equation}
\end{assumption}
\begin{remark}
  In the case $\mu_n=\nu_\beta$, where $\beta$ is as in Example
  \ref{ex011012-25}, one can easily see that
  $$
\int_{\Om_n}r_xr_{x'}\dd\mu_n=\int_{\Om_n}p_xp_{x'}\dd\mu_n=\beta_x^{-1}\delta_{x,x'},\quad x,x'\in\bbZ_n
$$
and Assumption \ref{ass2} holds. One can show, see Corollary
\ref{cor021212-25} below, that the bound \eqref{Hn2} persists
for any $t>0$.
  \end{remark}

\section{Main results}

 \subsection{Preliminaries}
 Our main result deals with the existence of the   limit
 $\lim_{n\to+\infty}{\cal E}_n[t;\varphi]$  of the
 emprirical distribution of the averaged microscopic energy density
 functional
 \begin{align}
   \label{Etp}
 {\cal E}_n[t;\varphi]:=  \frac{1}{n}\sum_{x=0}^{n}  \varphi\left(\frac x{n} \right)
 \bbE_{n}\big[ {\cal E}_{n,x}( t)\big] ,
\end{align}
  for any $\varphi\in C[0,1]$.
Before its rigorous formulation  we need some
preparatory material  that enables us to describe the limit.

% For any $\varphi\in L^2[0,1]$ consider its Fourier decomposition in
% the orhonormal cosine base, i.e.
% \begin{align*}
%   &\varphi(u)=\sum_{\ell=0}^{+\infty}\hat\varphi_c(\ell)c_\ell(u),\quad
%     \mbox{where}\\
%   &
%     \hat\varphi_c(\ell):=\int_0^1 \varphi(u) c_\ell(u)\dd u\quad\mbox{and}\\
%     & 
%  c_0(u)\equiv 1,\quad   c_\ell(u)=\sqrt{2}\cos(\pi \ell u), \,\ell=1,2,\ldots .
% \end{align*}

% \subsubsection{Discrete Fourier transform}
%  For $\varphi:[0,1]\to\bbR$ and $j=1,\dots,n$ we define the discrete sine
%  and cosine Fourier transforms
% \begin{equation}
%   \label{eq:18}
%   \begin{split}
% &  \widehat{\varphi}_{n,o}(j) := \frac{\sqrt 2}{n+1} \sum_{x=1}^n \sin(\pi j u_x) \varphi (u_x),\\
%   & \widehat{\varphi}_{n,e}(j) := \frac{\sqrt 2}{n+1} \sum_{x=0}^n \cos(\pi
%   j u_x) \varphi (u_x).
%   \end{split}
% \end{equation}
% For $j=0$ we let
% \begin{equation}
%   \label{eq:33}
%   \widehat{\varphi}_{n,e}(0) := \frac{1}{n+1} \sum_{x=0}^n \varphi (u_x).
% \end{equation}

  \subsubsection{Fractional Laplacian and the Green's function}
For any $\varphi\in L^2[0,1]$ consider its Fourier decomposition in
the orhonormal cosine base, i.e.
\begin{align*}
  &\varphi(u)=\sum_{\ell=0}^{+\infty}\hat\varphi_c(\ell)c_\ell(u),\quad
    \mbox{where}\\
  &
    \hat\varphi_c(\ell):=\int_0^1 \varphi(u) c_\ell(u)\dd u\quad\mbox{and}\\
    & 
 c_0(u)\equiv 1,\quad   c_\ell(u)=\sqrt{2}\cos(\pi \ell u), \,\ell=1,2,\ldots .
\end{align*}

For any $\alpha\in\bbR$ define the fractional $\alpha$-power of the
Neumann Laplacian on [0,1] as the operator
$-|\Delta|^{\al}:{\cal D}\big(|\Delta|^{\al}\big)\to L^2[0,1]$ given by
\begin{align}
  \label{fDlt}
  &-|\Delta|^{\al}\varphi(u)=
    -\sum_{\ell=0}^{+\infty}(\pi\ell)^{2\al}\hat\varphi_c(\ell)c_\ell(u),\quad\mbox{with}
  \\
&
 {\cal D}(|\Delta|^{\al})=  \Big[\varphi\in L^2[0,1]:\, \sum_{\ell=0}^{+\infty}(\pi\ell)^{4\al}
   \hat\varphi_c^2(\ell)<+\infty\Big].\notag
\end{align}
%In particular we let $\Delta =-|\Delta|^{1}$.

%  Suppose that $\varphi\in C^\infty_c(0,1)$.  By \cite[Lemma B.1]{KLO23}, for any $k>0$ we have for some constant $C>0$:
% \begin{equation}
%   \label{eq:19}
%   |\widehat{\varphi}_{n,\iota}(j)| \le  \frac{C}{\chi_n^k(j)},\quad j\in\bbZ,\,n=1,2,\ldots, \qquad \iota = o, e,
% \end{equation}
% where $\chi_n$ is $2n+2$-periodic extension of the function
% $$
% \chi_n(j)=(1+j)\wedge (2n+2-j), \quad j=0,\ldots,2n+1.
% $$
% In addition, if $\kappa\in(0,1)$, then there exists $C>0$ such that
% \begin{equation}
%   \label{eq:18a}
%  \sup_{|j|\le n^{\kappa}}\big(|\widehat{\varphi}_{n,o}(j)-\hat \varphi_s(j)|
%  +|\widehat{\varphi}_{n,e}(j)-\hat \varphi_c(j)|\big)\le \frac{C}{n^{1-\kappa}}.
% \end{equation}}

%   \subsubsection{Fractional Laplacian and the Green's function}

% For any $\alpha\in\bbR$ define the fractional $\alpha$-power of the
% Neumann Laplacian $\Delta_N$, as the operator $-|\Delta_N|^{\al}:{\cal
%   D}\big(|\Delta_N|^{\al}\big)\to L^2[0,1]$ given by
% \begin{align*}
% &-|\Delta_N|^{\al}\varphi(u)=-\sum_{\ell=0}^{+\infty}(\pi\ell)^{2\al}\hat\varphi_c(\ell)c_\ell(u),\quad\mbox{with}\notag\\
% &
% {\cal D}(|\Delta_N|^{\al})=\Big[\varphi\in L^2[0,1]:\, \sum_{\ell=0}^{+\infty}(\pi\ell)^{4\al}\hat\varphi_c^2(\ell)<+\infty\Big].
% \end{align*}
% In particular we let $\Delta_N=-|\Delta_N|^{1}$.

Suppose that $\al>0$. Let  $C^\infty[0,1]$ (resp. $C^\infty_c(0,1)$) be the set of all
$C^\infty$ smooth functions on $[0,1]$  (resp. compactly supported in $(0,1)$
smooth functions)
Define $H^{\alpha}[0,1]$ (resp. $H^{\alpha}_0[0,1]$)  as the
completion of $C^\infty[0,1]$ (resp. $C^\infty_c(0,1)$)
     under that norm
    \begin{equation}
      \label{Hal}
      \begin{split}
       & \|\varphi\|_{\alpha}:=\left(\|\varphi\|^2_{L^2[0,1]}+\|\varphi\|_{\alpha,0}^2\right)^{1/2}\quad\mbox{where}\\
       &
       \|\varphi\|_{L^2[0,1]}=\left(\int_0^1\varphi^2(u)\dd u\right)^{1/2}=\left(\sum_{\ell=0}^{+\infty}
       \hat\varphi_c(\ell) ^2\right)^{1/2},\\
       &
       \|\varphi\|_{\alpha,0}:=\Big(\sum_{\ell=1}^{+\infty}(\pi
       \ell)^{2\alpha} \hat\varphi_c^2(\ell) \Big)^{1/2}.
\end{split}
\end{equation}

Suppose that $\la>0$. The Green's function   is defined as 
\begin{align}
  \label{G-la}
&G_\la(u,v)=
                 \sum_{\ell=0}^{+\infty}\frac{c_\ell(u)
                 c_\ell(v)}{\la+(\pi\ell)^2}.
\end{align}
Let
 \begin{align}
  \label{V-la}
   &V_\la(u,v)=\la G_\la(u,v).
     \end{align}
     One can easily check that
    \begin{align*}
 \int_0^1V_\la(v,u)\dd u=   \int_0^1V_\la(u,v)\dd u=1.
    \end{align*}
    Using the method of images one can conclude that
    \begin{align}
      \label{G-la1}
&G_\la(u,v)
    =\sum_{n=-\infty}^{+\infty}\Big[G_\la(u-v+2n)+G_\la(2n+u+v)\Big],\quad
                 u,v\in[0,1],\quad\mbox{where}\\
  &
    G_\la(u)=\frac{1}{\sqrt{4\pi }}\int_0^{+\infty}e^{-\la
    t}e^{-u^2/(4t)}\frac{\dd t}{\sqrt{t}},\quad u\in\bbR.\notag
\end{align}
In particular, the above implies that $V_\la(v,u)\ge0$.

 \subsubsection{Fractional diffusion equation with Dirichlet boundary conditions}
\label{fde}

We denote by $L^2_w[0,1]$ the corresponding $L^2[0,1]$ space
equipped with the weak topology.

\begin{definition}
  \label{df1.5}
  Suppose that $c_{\rm bulk},c_{\rm bd}>0$, $T_0,T_1>0$ and $T_{\rm
    ini}\in L^2[0,1]$.
  We say that a 
  function $T:[0,+\infty)\to L^2[0,1]$
  is a weak solution of 
\begin{equation}
    \label{f-diff1}
    \begin{split}
      \partial_t T(t, u)&=-c_{\rm bulk}|\Delta|^{3/4}T(t, u)   \\
      &
  +c_{\rm  bd}    \sum_{v=0,1} \int_0^{+\infty}  \Big\{ {V_{\varrho}(u,v)}
 \int_0^1 V_{\varrho}(u',v)[ T_v -  T(t, u') ]\dd
  u'\Big\}\frac{\dd\varrho}{\varrho^{3/4}} ,
      \end{split}
    \end{equation}
with  the boundary
      values $T(t,v)=T_v$, $v=0,1$,  
    if the following   hold:
    \begin{itemize}
    \item[i)] $T\in C\Big([0,+\infty); L^2_w[0,1]\Big)$.
      % , where
      %   $L^2_w[0,1]$ is equipped with the weak topology,
      \item[ii)] for any $t>0$ and $v=0,1$ we have
        \begin{equation}
          \label{012111-24}
                      \int_0^t\dd s \int_0^{+\infty} 
 \Big( \int_0^1 V_{\varrho}(u',v)  \big(T_v -  T(s,u')\big)\dd u'
 \Big) ^2 \frac{\dd \varrho}{\varrho^{3/4}}<+\infty,
          \end{equation}
        \item[iii)] for any $\varphi\in C^\infty_c(0,1)$ we have
     \begin{equation}
    \label{042209-24aa}
    \begin{split}
    &  \langle\varphi,
  T(t )\rangle_{L^2[0,1]} -\langle\varphi,
  T_{\rm ini}\rangle _{L^2[0,1]}
    = - c_{\rm bulk}
 \int_0^t \langle|\Delta|^{3/4}\varphi, T(s )\rangle_{L^2[0,1]}  \dd s\\
     &
      +c_{\rm
    bd} \sum_{v=0,1} \int_0^t\dd s \int_0^{+\infty} \langle V_{\varrho}(\cdot
  ,v), \varphi\rangle_{L^2[0,1]}     
   \langle V_{\varrho}(\cdot
  ,v), T_v - T(s)\rangle_{L^2[0,1]}    \frac{\dd\varrho}{\varrho^{3/4}}.
    \end{split}
  \end{equation}
  \end{itemize}
  % Here,
%  % ${\rm m}_{[0,1]}$ is the Lebesgue measure on $[0,1]$ and
%   for a given function $\varphi$ and a Borel (signed) measure
% $\mu$ we let
% $
% \langle \varphi, \mu\rangle:= \int_0^1\varphi(u)\mu(\dd u) .
%  $ If $\mu$ is absolutely continuous with respect to the Lebesgue measure with the density $f$, then we write $
% \langle \varphi, f\rangle:= \int_0^1\varphi(u)f(u)\dd u.
%  $
\end{definition}

 \begin{theorem}
    \label{lm032209-24}
   Suppose that $T_{\rm ini}\in  L^2[0,1]$. Then,  equation
   \eqref{f-diff1} has a unique solution $T(\cdot,\cdot)$. In addition, the
   solution satisfies  
   \begin{equation}
     \label{boundary}
     \begin{split}
       &\int_0^tT(s,\cdot)\dd s\in C[0,1]\quad\mbox{and}\\
     &\int_0^tT(s,0)\dd s= {T_0} t,\quad  \int_0^tT(s,1)\dd s={T_1} t,\quad
     t\ge0.
     \end{split}
     \end{equation}
        \end{theorem}
The proof of Theorem \ref{lm032209-24} is presented in \cite[Section B of the Appendix]{KO-supp}.

\begin{remark}

  Using \eqref{G-la1}, \cite[Theorem 1.1 and Definition 2.4]{KW17} we
 conclude that for   $\varphi\in C^\infty[0,1]$ (the integral is understood in the principal value sense)
\begin{equation}
  \label{011504-25}
  \begin{split}
 &   |\Delta|^{3/4}\varphi(u)=
 \int_0^1 q(u',u) [\varphi(u')-\varphi(u)]
 \dd u',\\
 &
 \int_0^{+\infty}   V_{\varrho}(u
  ,v)  V_{\varrho}(u'
  ,v) \frac{\dd\varrho}{\varrho^{3/4}}=g(u,u';v),
 \end{split}
\end{equation}
 with 
  \begin{align}
    \label{012412-24}
    &q(u,u'):=\frac{ 3}{2^{5/2} {\pi}^{1/2}} \sum_{n\in\bbZ} 
        \Big(\frac{1 }{|u+u'+2n|^{5/2}} +\frac{1}{|u-u'+2n|^{5/2}} \Big) ,\notag\\
    &g(u,u';v)=\sum_{n,n'\in\bbZ}W (u+v+2n,u'+v+2n'),\quad v=0,1,\quad \mbox{where}\notag\\
    &
      W(u,u'):=  \frac{5\Gamma\Big(\frac{1}{4}\Big)^2}{2^{5}\pi}\int_0^{\pi/2}
      \Big(\frac {\sin^2(2\theta)}{ (u\sin\theta)^2+(u'\cos\theta)^2}\Big)^{5/4}\dd \theta
.  
  \end{align}
{Here $\Gamma(\cdot)$ is the Euler gamma function.}

  Obviously $W(u,u')=W(u',u)$ {and an elementary calculation leads
    to 
    \begin{equation}
      \label{buv}
   \int_0^1g(u,u';v) \dd u'  
    = \sqrt{\pi} \sum_{n\in\bbZ} \frac{1}{|u+v+2n|^{3/2}} ,\qquad
    v=0,1.   
  \end{equation}}

  \end{remark}

   {
     \begin{remark}
       We can rewrite \eqref{f-diff1} as, cf \eqref{buv},
        \begin{equation}
    \label{eq:71}
    \begin{split}
      \partial_t T(t, u)&= -c_{\rm bulk}|\Delta|^{3/4}T(t, u)  + c_{\rm  bd}    \sum_{v=0,1}  
 \int_0^1 g(u,u';v) [ T_v -  T(t, u) ]\dd u'\\ 
&\qquad +
 c_{\rm  bd}    \sum_{v=0,1}
 \int_0^1 g(u,u';v) [ T(t, u) -  T(t, u') ]\dd u'\\
& = \int_0^1 r(u,u') \left[T(t,u') - T(t,u)\right] du'
    + \sum_{v = 0,1} b(u;v) \left[T_v - T(t,u)\right],
    \end{split}
  \end{equation}
     where
     \begin{equation}
       \label{eq:5}
       \begin{split}
         & r(u,u') := c_{\rm bulk} q(u',u) - c_{\rm bd} \sum_{v=0,1}
         g(u,u'; v),\\
         &
       b(u;v):=  c_{\rm  bd}     
 \int_0^1 g(u,u';v)  \dd u'.
         \end{split}
       \end{equation}
      
   \end{remark}

   \begin{remark}
     \label{rem:probint}
     If $r(u,u')\ge 0$ we can interpret  \eqref{eq:71} as the
       equation describing  the evolution of the density $T(t,u)$ of
       a Markov process
     with creation and annihilation. The dynamics of the
     process can be described as follows:
     {a particle}
     jumps from $u$ to $u'$ with rate $r(u,u')$ (this takes into account the
     jumps with reflection of the fractional Laplacian minus
     {the jumps censored by the boundaries}).
     At time $t$ and position $u$ the particle gets annihilated at the rate
     {$  b(u,0)+ b(u,1) $}
     and it is created at this site at  the
     rate $ b(u,0) T_L +  b(u,1) T_R$. 
   \end{remark}
}

\subsection{Formulation of the main result}

\begin{theorem}[The limit of thermal energy and equipartition]
  \label{thm012911-23}
  Under the assumptions about the initial data made in Section
  \ref{sec1.4}, for any continuous test function $\varphi:[0,1]\to\bbR$ and any $t\ge 0$, we have
\begin{align}
   \lim_{n\to+\infty}\frac{1}{n}\sum_{x=0}^{n}  \varphi\left(\frac x{n} \right)
 \bbE_{n}\big[ {\cal E}_{x}( t)\big] =\int_0^1 T(t,u)\varphi(u)\dd u,
 \label{eq:conv-temp}
\end{align}
where $T(t,u)$ is the solution of \eqref{f-diff1} with the initial data 
$T(0,u)=T_{\rm ini}(u)$ and the boundary conditions $T(t,0)=T_L$,
$T(t,1)=T_R$. Here
\begin{align}
    \label{022111-24}
       & c_{\rm bulk} =\frac{1}{(2^3\ga)^{1/2}}
       ,\quad
         c_{\rm bd}  =\frac{ \sqrt{2}\tilde\ga c_{\rm
         bulk}}{(1+\tilde\ga)^2}. 
          \end{align}
In addition, for any compactly supported,
continuous function $\Phi:\bbR_+\times [0,1]\to\bbR$
\begin{equation}
  \label{eq:conv-temp2}
  \begin{split}
  \lim_{n\to+\infty}\frac{1}{n}\sum_{x=0}^{n} \int_{0}^{+\infty} \Phi\left(t,\frac x{n} \right)
 \bbE_{n}\big[ p_{x}^2( t)\big] \dd t
  &=
   \lim_{n\to+\infty}\frac{1}{n}\sum_{x=0}^{n} \int_{0}^{+\infty} \Phi\left(t,\frac x{n} \right)
   \bbE_{n}\big[ {\cal E}_{x}( t)\big] \dd t
   \\
   &=  \int_{0}^{+\infty} \dd t\int_0^1 T(t,u) \Phi\left(t, u\right)\dd u.
 \end{split}
\end{equation}
\end{theorem}
  In what follows we present an outline of the  
  proof  of the result. The detailed argument is given in \cite{KO-supp}.
 
\begin{remark}
  Concerning the coefficients $c_{\rm bulk} $ and $c_{\rm bd}$, note that
\begin{align*}
                &
            \frac{c_{\rm bd}}{ c_{\rm bulk}}\to0,\quad\mbox{as}\quad 
             \tilde\ga\to0,\quad\mbox{or}\quad \tilde\ga\to+\infty,
\end{align*}
In addition the ratio  $\frac{c_{\rm bd}}{ c_{\rm bulk}}$ achieves
maximum  $\frac{1}{2^{3/2}}$ at
$\tilde\ga_{\rm max}=1$.
  \end{remark}

  Finally concerning the estimates on the size of the current, the
  following result holds, see \cite[Theorem 2.14]{KO-supp} 
  \begin{theorem}[Bound on  the energy currents]
    \label{thm-current}
  For any $t_*>0$  there exists $C_{\mathcal{J},t_*}>0$
  such that
  \begin{equation*}
      \sup_{x=0,\ldots,n+1}|\int_0^t\bbE_n\left[ j_{x-1,x}( s)\right] \dd s|\le
   \frac{C_{\mathcal{J},t_*}}{\sqrt{n}},\quad t\in[0,t_*],\,n=1,2,\ldots.
        \end{equation*}
      \end{theorem}

     % \textcolor{red}{\bf DOTAD}
 
\section{Entropy and energy bounds}

     \label{sec3}
     
\begin{theorem}[Entropy bound]
  \label{entropy-t}
  Under the assumption made in Section \ref{sec1.4},  for any $t_*>0$
   there exists a constant $C_{H,t_*}>0$ such that
   \begin{align}
     \label{011212-25}
        {\mathbf{H}}_{n,T}(t) 
  \le C_{H,t_*} n,\qquad t\in[0,t_*].
  \end{align}
\end{theorem}
\proof
In the case $T_L=T_R=T$ the measure  $\nu_T$ is stationary for the
process, consequently the relative entropy w.r.t.  the measure
$ {\mathbf{H}}_{n,T}(t)$ is not increasing in $t$. In this case
the case appearing in \eqref{011212-25} shall not depend on $t_*$. The proof in
the general case $T_L,T_R>0$ is given in \cite[Section 12]{KO-supp}. 
\qed

Let
      \begin{align*}
  \mathcal{H}_n (\mathbf r, \mathbf p):=
  \sum_{x=0}^n {\cal E}_x (\mathbf r, \mathbf p),\quad   \mathcal{H}_n (t):= \mathcal{H}_n \big(\mathbf r(t), \mathbf p(t)\big).
      \end{align*}
Using the entropy inequality, see \cite[p. 338]{kl}, we conclude that
for any $\alpha>0$
  \begin{equation}
\label{entropy-in-tz}
\begin{split}
  &\mathbb E_{\mu_n} \big[ \mathcal{H}_n (t) \big] =\int_{\Om_n}
\Big(\sum_{x=0}^n{\cal E}_x\Big)   f_n(t) \dd
 \nu_{T}
 \\
 &
 \le \frac{1}{\al}\left\{\log\bigg(\int_{\Om_n}
     \exp\left\{ \frac\al 2 \sum_{x=0}^n(p_x^2+r_x^2)
     \right\}\dd  \nu_{T} \bigg)+ {\mathbf{H}}_{n}(t)\right\}.
\end{split}
\end{equation}
Choosing $\al<T^{-1}$ we conclude that there exist constants
  $ C,C'>0$ such that
 \begin{equation*}
  \bbE_n\left[{\cal  H}_n(t)  \right]
   \le C\big(n+ {\bf  H}_{n }(t)\big)  \le C'n,\qquad t\ge0,\,n=1,2,\ldots.
 \end{equation*}
Hence the following holds.

\begin{corollary}[Energy bound]
  \label{cor011212-25}
  For any $t_*>0$  there exists $C_{\mathcal{H},t_*}>0$ such that
\begin{equation}\label{eq:energyb}
 \mathbb E_{n}  \big[\mathcal{H}_n (t)\big]
  \le C_{\mathcal{H},t_*} n,\qquad t\in[0,t_*],\,n=1,2,\ldots.
\end{equation}
\end{corollary}

% \begin{corollary}
%   \label{cor011212-25}
%   Suppose that $T_L,T_R>0$. Then,  for any $t_*>0$  there exists $C_{ t_*}>0$
%   such that
%   \begin{equation}
%   \label{eq:9z}
% \bbE \left[\sum_{x=0}^n{\cal E}_x( t)\right]
%  \le  C_{t_*} n ,\quad t\in[0,t_*],\,n=1,2,\ldots.
% \end{equation}
% \end{corollary}

\subsubsection*{Identity for covariances}

  After  tedious, but otherwise elementary calculations (that can be
  found in \cite[Proposition 3.2]{KO-supp}, see also   Supplement 1
  ibid., for detailed calculations), we conclude
  the following identity (see \eqref{Hn2a})
  \begin{align}
    \label{021212-25}
  & {\cal H}_n^{(2)}(t) + {\frak B}_n(t)
    +{\frak b}_n(t) ={\cal H}_n^{(2)}(0) \\
&
 +\frac{2\tilde \ga
     n^{3/2}}{n+1}\int_0^t\left[T_L\Big(T_L-\mathbb E_n [p^2_0(  s)]
                    \Big) 
    +  T_R\Big( T_R-\mathbb E_n[ p_n^2(  s)] \Big)\right]\dd s,\notag
\end{align}
where
\begin{align*}
  &  {\frak B}_n(t)=\frac{2\tilde \ga n^{3/2}}{n+1}\sum_{v=0,1}\int_0^t\Big\{T_v-\mathbb E_n\left[
                p^2_{nv}( s)
                  \right] \Big\}^2\dd s 
 \\
&+ \frac{2\tilde \ga n^{3/2}}{n+1} \sum_{v=0,1}\sum_{x=1}^n
    \int_0^t\left\{\mathbb E_n\left[ p_{nv}( s)
  p_{x}( s)\right] \right\}^2 \dd s 
 \\
 &
 + \frac {2 \tilde \ga n^{3/2}}{n+1} \sum_{v=0,1}\sum_{x'=1}^n   \int_0^t\left\{ \mathbb E_n\left[ p_{nv}( s)
  r_{x'} (  s)\right] \right\}^2  \dd s 
\end{align*}
and 
\begin{align*}
 &
 {\frak b}_n(t) := \frac{\ga
    n^{3/2}}{n+1}\sum_{x=0}^{n-1}\int_0^t\Big[\nabla \mathbb
    E_n[ p_{x}^2( s)]
                   \Big]^2\dd s 
  \\
  &
    +\frac{2\ga
                       n^{3/2}}{n+1}\sum_{x=1}^{n}\mathop{\sum_{x'=0}^{n}}_{x'\not\in\{x-1,x\}}\int_0^t\left\{\mathbb
                       E_n\left[ \nabla^\star p_x( s)
                 p_{x'}( s)\right] \right\}^2 \dd s  \\
  &
    + \frac {2 \ga n^{3/2}}{n+1}\sum_{x=1}^{n} \sum_{x'=1}^n
    \int_0^t \left\{\mathbb E_n\left[ \nabla^\star p_x( s)
  r_{x'} ( s)\right] \right\}^2\dd s.\notag
\end{align*}
From the identity \eqref{021212-25} one can conclude the following.
\begin{theorem} \label{cor012102-24}
  Suppose that $T_L,T_R>0$. Then,  for any $t_*>0$  there exists $C>0$
  s.t.
  \begin{equation}
  \label{eq:9}
  {\cal H}_n^{(2)}(t)    \le C
\end{equation}
\end{theorem}
\proof We present the argument 
in the special case $T_L=T_R=T$. The proof of the general statement
can be found in \cite[Section 13.7]{KO-supp}.

Using \eqref{current} we obtain
\begin{equation*}
\frac{\dd}{\dd t}\bbE_n {\cal H}_n(t)= n^{3/2} 
\bbE_n\left(j_{-1,0} (t) - j_{n,n+1} (t)\right) \dd s.
\end{equation*}
 Combining  with \eqref{021212-25} we conclude  
   \begin{equation*}
     \begin{split}
   &    {\cal H}_n^{(2)}(t)  \le {\cal H}_n^{(2)}(0) + \frac{2 T n^{3/2}}{n+1}
   \int_0^t \bbE_n\left(j_{-1,0} (s) - j_{n,n+1} (s)\right)  ,\\
   &
  \le  {\cal H}_n^{(2)}(0)    + \frac{2 T \bbE_n {\cal H}_n(t)}{n+1}.
\end{split}
\end{equation*}
The conclusion of the theorem then follows from Corollary  \ref{cor011212-25}.\qed

\medskip

\begin{corollary}
  \label{cor021212-25}
 Under the assumptions about the initial data made in Section
  \ref{sec1.4} for  any $t_*\ge 0$ there exists a constant
  {$C_{t_*}>0$} such that
\begin{align}
    \frac{1}{n}\sum_{x=0}^{n}   
 \big[ \bbE_{n} {\cal E}_{x}( t)\big]^2 \le C_{t_*},\quad t\in[0,t_*],\,n=1,2,\ldots.
 \label{eq:L2}
\end{align}
  \end{corollary}

  \section{Limit identification}

  In light of the results of  Section \ref{sec3} (see in
  particular  Corollary \ref{cor021212-25}) one can
  conclude the compactness of the family of functionals 
$t\mapsto {\cal E}_n[t;\cdot]$, $n=1,2,\ldots$  (see \eqref{Etp}) in
the space $C\big([0,t_*];L^2_w[0,1])$, for any $t_*>0$.
  In the present section we outline the argument that allows to
  identify the set of limiting points of the family as the singleton
  described in Theorem \ref{lm032209-24}. It suffices to show
  therefore that for any test function
  $\varphi\in C^\infty_c(0,1)$ we have
  \begin{equation}
    \label{011212-25e}
    \lim_{n\to+\infty}{\cal E}_n[t;\varphi]=
    { \int_0^1T(t,u)\varphi(u)\dd u,}
    \end{equation}
where  $T(t,u)$ satisfies equation
   \eqref{f-diff1} in the sense of Definition \ref{df1.5}. To simplify
   the notation we shall write 
   \begin{align*}
\varphi_x = \varphi(u_x), \quad \varphi_x' = \varphi'(u_x)\quad\mbox{and}\quad   u_x=\frac{x}{n+1}. %  \quad\mbox{and}
%\varphi_{n,x}' :=n(\varphi_{x+1} -\varphi_{x})
\end{align*}
After using \eqref{current} and  performing  the summation by parts we obtain
\begin{equation*}
  \begin{split}
    {\rm E}_n [t, \varphi] -  {\rm E}_n [0, \varphi]  
  = \frac{n^{1/2}}{n+1} \sum_{x=0}^{n-1}  \varphi_{x} '
  \int_0^{t} \mathbb E_n\left[ j_{x,x+1} ( s) \right] \dd s+o_n(1).
\end{split}
\end{equation*}
                                                                        Both
                                                                        here
                                                                        and
                                                                        in
                                                                        what
                                                                        follows
                                                                        we
                                                          have $\lim_{n\to+\infty}o_n(1)=0$.

     \subsection{Current decomposition}
                                                                        Hence,
   \begin{align}
   \label{021312-25}
   {\cal E}_n  [t;\varphi] -   {\cal E}_n  [0; \varphi] &= J_n ^{(a)} (t;\varphi')+
    J_n^{(s)}(t;\varphi')+o_n(1),\quad\mbox{where} \\  
    J_n ^{(a)} (t; \varphi') &:=
    \frac 1{\sqrt n} \sum_{x=1}^n \varphi'_{x}\int_0^t\bbE_n\left[ j_{x-1,x}^{(a)}(s)\right] \dd s
 \notag \\
  &= -\frac 1{\sqrt n} \sum_{x=1}^n \varphi'_{x}  \int_0^{t} \mathbb
    E_n\left[   p_{x-1}(s)r_{x}(s)\right] \dd s 
    ,\notag\\
       J_n^{(s)}(t;\varphi') &:=  - \frac{\ga}{2\sqrt n} \sum_{x=0}^{n-1}  \varphi_{x} '
    \int_0^{t} \mathbb E_n\left[ \nabla p_x^2(s)\right] \dd s. \notag
\end{align}
Summing by parts one more time and using energy bound
\eqref{eq:energyb}
we conclude that the input of the stochastic part of the current is
negligible, i.e.
\begin{align*}
  &
  J_n^{(s)}(t;\varphi')  =\frac{\ga }{2n^{3/2}} \sum_{x=0}^{n-1}  \varphi_{x} ''
    \int_0^{t} \mathbb E \left[p_x^2( s) \right] \dd s
    +o_n(1)=o_n(1).  
\end{align*}
{Here  $\varphi''_x:=\varphi''(u_x)$.}
{Consequently we have
\begin{equation}
  \label{eq:Ja}
  {\cal E}_n  [t;\varphi] -   {\cal E}_n  [0; \varphi] = J_n ^{(a)} (t;\varphi') + o_n(1).
\end{equation}
To deal with  the Hamiltonian current $J_n ^{(a)} (t; \varphi')$ we
wish to express the integral of the expectation of
the covariance field $\int_0^{t} \mathbb E_n\left[ p_{x-1}(s)r_{x}(s)\right] \dd s$
in terms of the  expectation of the  microscopic energy density
$ \bbE_n\left[ {\cal E}_x(s) \right]$. For  that purpose
we study the covariances of the momenta
    and stretches using the fact that their dynamics is given by a
  linear  system of stochastic differential equations with a
  multiplicative noise, see \eqref{eq:bas1}. }

\subsection{Matrix of covariances}

% To deal with    the Hamiltonian current $J_n ^{(a)} (t; \varphi')$ we
% wish to express the integral of the expectation of  the covariance field $\int_0^{t} \mathbb
%     E_n\left[   p_{x-1}(s)r_{x}(s)\right] \dd s $ in terms of the
%     expectation of the 
%     microscopic energy density $ E_n\left[ {\cal E}_x(s) \right]$. For
%     that purpose we wish to determine the covariances of the momenta
%     and stretches using the fact that their dynamics is given by a
%   linear  system of stochastic differential equations with a
%   multiplicative noise, see \eqref{eq:bas1}.
  Define the $2\times 2$ block matrix of covariances 
\begin{equation*}
\label{S1ts}
S(t)  
=\left[
  \begin{array}{cc}
    {S^{(r)}(t)}&S^{(r,p)}(t)\\
   S^{(p,r)}(t)& S^{(p)}(t)
  \end{array}
\right],
\end{equation*}
where
\begin{align*}
&S^{(r)}(t)=\Big[\bbE_{n}[r_{x}( t)r_{y}(t)]\Big]_{x,y=1,\ldots,n},\quad S^{(r,p)}(t)=\Big[\bbE_{n}[r_{x}( t)p_{y}(t)]\Big]_{x=1,\ldots,n,y=0,\ldots,n},\notag\\
&\\
&
S^{(p)}(t)=\Big[\bbE_{n}[p_{x}( t)p_{y}(t)]\Big]_{x,y=0,\ldots,n}\quad \mbox{and}\quad S^{(p,r)}(t)=\Big[S^{(r,p)}(t)\Big]^T. \notag
\end{align*}
Define a column vector of  stretches and momenta
 $${\bf X}(t)=\begin{bmatrix} {\bf r}(t)\\
 {\bf p}(t)
\end{bmatrix}
$$
The system \eqref{eq:bas1}
 can be rewritten in the matrix form
\begin{equation*}
d{\bf X} (t)= -n^{3/2}A {\bf X}(t)dt+ 
 \Sigma \Big({\bf p} (t-)\Big)\dd M_n(t),
\end{equation*}
where $  A$ is a $2\times 2$- block matrix
\begin{equation*}
  A=
\left(
  \begin{array}{cc}
    0 &-\nabla^\star\\
   -\nabla& -\ga \Delta_{\rm N}+\tilde\ga  E
  \end{array}
\right),
\end{equation*}
where $\nabla$ is the $(n+1)\times n$ matrix corresponding to the
operator defined in \eqref{nbl},  $-\nabla^\star$ is its transpose,
$\Delta_{\rm N}$ is the 
 $(n+1)\times (n+1)$ matrix  of  the discrete Neumann Laplacian  and
$E=[\delta_{x,0}\delta_{y,0}+\delta_{x,n}\delta_{y,n}]_{x,y=0,\ldots,n}$.
Finally $M_n(t) $ is a zero mean
  vector, martingale
  $$
  M(s)^T =\left( 
                   0_{n,1},
                  n^{3/4}   w_L(s),
                    \tilde N^{(n)}_{0,1}(\gamma s),
                  \ldots,
                    \tilde N^{(n)}_{n-1,n}(\gamma s),
                  n^{3/4}    w_R(s)
                \right). $$
                For a given real, or complex valued function $f$
                defined on $[0,+\infty)$ we let
                \begin{align*}
                  &
       \delta_{0,t}f:=\bbE_nf(0)-\bbE_nf(t),\qquad
                   \lang f\rang_t:=\int_0^t\bbE_n f(s)\dd s.
                \end{align*}
 Using   the dynamics, see  \eqref{011012-25}, we obtain the equality 
 \begin{align}
   \label{011312-25}
              &A \lang S \rang_t +\lang S \rang_t A^T =
       \Sigma_2 \Big(\lang \overline{\frak {\bf (\nabla p)}^2}\rang_t
                \Big) + \frac{1}{n^{3/2}}\delta_{0,t} S,\quad
                \mbox{where  }\\   
              &
                \overline{\frak {\bf (\nabla p)}^2}(s)=\Big[\bbE_{n}\big(\nabla^\star
  p_1(s)\big)^2,\ldots, \bbE_{n}\big(\nabla^\star
  p_n(s)\big)^2\Big].\notag
           \end{align}

           \subsubsection{Fourier transforms of covariances}

      Consider the orthonormal bases of $\bbR^{n+1}$ and $\bbR^n$
      given by        
\begin{align*}
    & 
    \psi_j(x)=\left(\frac{2-\delta_{0,j}}{n+1}\right)^{1/2}\cos\left(\frac{\pi
    j(2x+1)}{2(n+1)}\right),\quad x,j=0,\ldots,n,\\
  &
    \phi_j(x)=\left(\frac{2}{n+1}\right)^{1/2}\sin\left(\frac{jx\pi}{n+1}\right) ,\quad x,j=1,\ldots,n.
\end{align*}
We have
\begin{align*}
  &\nabla^\star \psi_j=-\ga_j \phi_j\quad\mbox{and}\quad\nabla
    \phi_j=\ga_j \psi_j,\quad \mbox{where } \ga_j=   2\sin\left(\frac{j\pi}{2(n+1)}\right).
\end{align*}
Define the Fourier transforms of the stretches and momenta as
\begin{align*}
  &
     \tilde r_{j}(t):=\sum_{x=1}^n \phi_j(x)r_x(t)\quad \mbox{and}\quad
      \tilde p_{j}(t):=\sum_{x=0}^n \psi_j(x)p_x(t) .\end{align*}
We can furthermore define the Fourier transforms of the respective
covariance matrices as follows
\begin{align*}
&\tilde S^{(r)}(t)=\Big[\lang
                \tilde r_j\tilde r_{j'}\rang_t\Big]_{j,j'=1,\ldots,n},\quad \tilde S^{(rp)}(t)=\Big[\lang
                \tilde r_j\tilde p_{j'}\rang_t\Big]_{j=1,\ldots,n,j'=0,\ldots,n},\notag\\
&\\
&
\tilde S^{(p)}(t)=\Big[\lang
                \tilde p_j\tilde p_{j'}\rang_t\Big]_{j,j'=0,\ldots,n}\quad \mbox{and}\quad \tilde S^{(pr)}(t)=\Big[\tilde S^{(rp)}(t)\Big]^T,  
\end{align*}

\subsubsection{Resolution of the covariances}
Using \eqref{011312-25} we can find the formulas for the Fourier
transforms of the cavariance matrices in terms of $\overline{\frak
  {\bf (\nabla p)}^2}(s)$ and the covariances between the variables in
the bulk and on the boundaries. More precisely, (see \cite[Supplement 2]{KO-supp}) if 
   $I=\{p,r,pr,r\}$, then
\begin{align}
\label{031312-25}    \tilde S^{(\iota)}_{j,j'} &=
 \Theta_\iota(\la_j,\la_{j'})  F_{j,j'}
  +  \sum_{\iota'\in
     I}\Pi^{(\iota)}_{\iota'}(\la_j,\la_{j'}) B_{j,j'}^{(\iota')}
   +  \sum_{\iota'\in
                                I}\Xi^{(\iota)}_{\iota'}(\la_j,\la_{j'}) R_{j,j'}^{(\iota')}.
\end{align}
Here {$\la_j=\ga_j^2$}, 
$$
              F_{j,j'}=\ga \sum_{y=1}^n\phi_j(y) \phi_{j'}(y) \lang
              (\nabla^\star p_y)^2\rang_t .
              $$
The terms $ B_{j,j'}^{(\iota')} $ correspond to
the boundary-bulk covariances:
\begin{align}
  \label{bpr}
  &
    B_{j,j'}^{(pr)}=\psi_{j}(0) \tilde s^{(p,\tilde
    r)}_{0,j'}+\psi_{j}(n) \tilde s^{(p,\tilde r)}_{n,j'},\quad B_{j,j'}^{(rp)}=B_{j',j}^{(pr)},\\
 & \tilde s^{(p,\tilde r)}_{z,j}= \lang \tilde r_j p_z\rang_t=\int_0^t
   \tilde b^{(p,\tilde r)}_{z,j}(s)\dd s, \quad\mbox{where}\notag\\
                          &
    b^{(p)}_{z,x}(s)=\bbE_n\big[
      p_{z}( s) r_x( s)\big]  ,\quad    \tilde b^{(p,\tilde r)}_{z,j}(s)=  \bbE_n\big[ \tilde r_j(s) p_z(s)
\big] \notag
\end{align}
and 
\begin{align}
  \label{bpp}
    &
     B_{j,j'}^{(p)} =  B_{j,j'}^{(p,0)} +
     B_{j,j'}^{(p,n)},\quad\mbox{where}\\
&
    B_{j,j'}^{(p,z)} =  \psi_j(z)\lang\tilde b_{z,j'}^{(p)} \rang_t+ \lang\tilde b_{z,j}^{(p)} \rang_t\psi_{j'}(z)
                                            \quad\mbox{and}\notag\\
    &
      \tilde b_{z,j}^{(p)}(s) := \sum_{x=0}^nb^{(p)}_{z,x}(s)\psi_j(x), \notag\\
    &
      b^{(p)}_{0,0}(s)= T_L-\bbE_n 
      p_{0}^2(s)  , \quad b^{(p)}_{n,n}(s)= T_R-\bbE_n 
      p_{n}^2(s) ,\notag\\
    &
    b^{(p)}_{z,x}(s)=- \bbE_n\big[
      p_{z}( s) p_x( s)\big]  ,\quad x\not=z. \notag
  \end{align}
  The  
  fluctuations of the covariances are given by
  \begin{equation}
    \label{041312-25}
    R_{j,j'}^{(\iota)}=\frac{1}{n^{3/2}} \delta_{0,t}{\tilde
      S}_{j,j'}^{(\iota)}.
  \end{equation}

  Finally, for $c,c'\in[0,4]$ the  coefficients  $\Theta^{(\iota)}_{\iota'}(c,c')$,
 $\Pi^{(\iota)}_{\iota'}(c,c')$ and $\Xi^{(\iota)}_{\iota'}(c,c')$  are
 given by  
  \begin{equation}
    \label{Theta}
    \begin{split}
   &\Theta_p(c,c')=\frac {2\ga
     cc'}{\theta(c,c')},\quad\mbox{where}\quad 
     \theta(c,c')= (c-c')^2+ 2\ga^2cc'(c+c') ,\\
     &
     \Theta_r(c,c') =\frac{\ga (c+c')
   \sqrt{c c'}}{\theta(c,c')},\quad  
\Theta_{pr}(c,c')=\frac
                 { (c-c' )\sqrt{c'}}{\theta(c,c')},
\end{split}
\end{equation}
   \begin{equation}
    \label{Xi}
    \begin{split}
    &
      \Xi^{(\iota)}_p(c,c')= \Theta_\iota(c,c'),\quad \iota=p,pr,r,\\
      &
      \Xi^{(p)}_{r}(c,c')=-\Xi^{(pr)}_{rp}(c,c') =\Theta_r(c,c'),\\
       &
        \Xi^{(p)}_{pr}(c,c')=-\Theta_{pr}(c,c'),\quad   \Xi^{(p)}_{rp}(c,c')=\Xi^{(p)}_{pr}(c',c),
\\
  &
       \Xi^{(pr)}_{pr}(c,c') =\frac
    {  \ga c'(c+c')}{\theta(c,c')} , \\
                  &
    \Xi^{(p,r)}_{r}(c,c') =\frac
                 { 1}{2 \sqrt{c} }\Big[1+\frac{c^2-(c')^2
         }{ \theta(c,c')}\Big],
  \\
  % &
  %      =\frac{ \sqrt{c}\Big[c-c'+\ga^2c'(c+c')\Big]}{ \theta(c,c') },\\
         &
 \Xi^{(r)}_{rp}(c,c')=   \Xi^{(r)}_{pr}(c',c)=-\Xi^{(pr)}_{r}(c,c') ,\\
    &
    \Xi^{(r)}_{r}(c,c')=\ga\frac{
    c^2+(c')^2+\ga^2cc'(c+c')}{\theta(c,c')}.
\end{split}
\end{equation}
 and
\begin{equation}
    \label{Pi}
    \begin{split}
      &\Pi^{(\iota)}_p(c,c')=\tilde \ga\Theta_\iota(c,c'),\quad \iota=p,pr,r,\\
      &
      \Pi^{(\iota)}_{\iota'}(c,c')=-\tilde \ga  \Xi^{(\iota)}_{\iota'}(c,c'),\quad
   \iota=p,pr,r,\,\iota'=pr,\,rp,\\
    &   \Pi^{(\iota)}_{r}(c,c')=0,\quad \iota=p,pr,rp,r.
\end{split}
\end{equation}

 \subsection{Bulk-boundary current decomposition}
\label{sec:bulk-bound-curr}

Going back to \eqref{eq:Ja}, we can rewrite   it, using the
formula for the expectation of the Hamiltonian current
$\bbE_n\big[p_x(s)r_{x+1}(s)\big]$ that follows from \eqref{031312-25}.
 Then, 
 \begin{equation}
   \label{eq:energy}
    \begin{split}
 {  J_n ^{(a)} (t;\varphi')}
    = -\frac 1n
      \theta_{pr}(\varphi';n) 
    -   \frac 1n \sum_{\iota\in
        I}\pi^{(pr)}_{\iota}(\varphi';n)    - \frac 1n
    \sum_{\iota\in I}\xi^{(pr)}_{\iota}(\varphi';n).
\end{split}
\end{equation}
Here $I=\{p,pr,rp,r\}$ and 
\begin{align}
 & \theta_{pr}(\varphi';n)= \sum_{j,j'=1}^{n} {\cal W}_{j,j'}
   \sqrt{\la_j\la_{j'}}  \Theta_{pr}  (\la_j,\la_{j'})    F_{j,j'}, \qquad \text{(bulk term)} \label{eq:theta}\\
  &
    \pi^{(pr)}_{\iota} (\varphi';n)
  =     \sum_{j=0}^{n}\sum_{j'=1}^{n} {\cal W}_{j,j'} \Pi^{(pr)}_{\iota} 
    (\la_j,\la_{j'})B_{j,j'}^{(\iota)},\qquad \text{(boundary terms)}\label{eq:pi} \\
     &
  \xi^{(pr)}_{\iota}(\varphi'; n) = \sum_{j=0}^{n}\sum_{j'=1}^{n} {\cal W}_{j,j'}
      \Xi^{(pr)}_\iota(\la_j,\la_{j'})
       R_{j,j'}^{(\iota)},\qquad \text{(time boundary terms)}  \label{eq:xi}\\
  &
  {\cal W}_{j,j'}:= \sqrt n \sum_{x=1}^{n} 
    \phi_{j'}(x) \psi_{j}(x-1)  \varphi' (u_{x})  \label{eq:Wjj}. 
\end{align}
 
Thanks to the presence of factor $n^{-3/2}$ in \eqref{041312-25} and
the energy bound \eqref{eq:energyb},
it can be shown that   the {time boundary} terms   $\sum_{\iota\in
  I}\xi^{(pr)}_{\iota}(\varphi';n)$   vanish, as
$n\to+\infty$, see \cite[Section 10]{KO-supp} for details.

\subsection{{Asymptotics for the bulk terms
  $\theta_{pr}(\varphi';n)$}}

% Concerning % the term $\theta_{pr}(\varphi';n)$, dealing with
% the bulk terms 
% we have the following.
\begin{proposition}
  \label{lem-pr}
  For any $\varphi\in C_c^\infty(0,1)$ we have, as $n\to+\infty$ {(cf \eqref{fDlt})},
    \begin{equation}
  \label{051712-24}
  \begin{split}
    \frac{1}{n}  \theta_{pr}(\varphi';n)  
    &=  \frac{1}{(2^3 \gamma)^{1/2}n}
      \sum_{y=1}^n 
     \int_0^t{\cal E}_y(s)\dd s
      \sum_{\ell=1}^{+\infty} (\pi\ell)^{3/2}\hat{\varphi}_c(\ell) c_\ell(u_y) + o_n(1)\\
      & =  \frac{1}{(2^3 \gamma)^{1/2}}
      \int_0^t {{\mathcal E}_{n} \left[s; |\Delta|^{3/4} \varphi \right]} \dd s
      + o_n(1).
    \end{split}
  \end{equation}
  \end{proposition}

%\textcolor{red}{\bf DOTAD}

 \subsubsection{Outline of the proof of  Proposition \ref{lem-pr} }
  
%\subsubsection{Discrete Fourier transform}
 For $\varphi:[0,1]\to\bbR$ and $j=1,\dots,n$ we define the discrete sine
 and cosine Fourier transforms
\begin{equation}
  \label{eq:18}
  \begin{split}
&  \widehat{\varphi}_{n,o}(j) := \frac{\sqrt 2}{n+1} \sum_{x=1}^n \sin(\pi j u_x) \varphi (u_x),\\
  & \widehat{\varphi}_{n,e}(j) := \frac{\sqrt 2}{n+1} \sum_{x=0}^n \cos(\pi
  j u_x) \varphi (u_x).
  \end{split}
\end{equation}
For $j=0$ we let
\begin{equation}
  \label{eq:33}
  \widehat{\varphi}_{n,e}(0) := \frac{1}{n+1} \sum_{x=0}^n \varphi (u_x).
\end{equation}

  Denote $ k_j:=\frac{j}{n+1}$.
% \begin{equation}
%   \label{k-j}
%   k_j:=\frac{j}{n+1}.
%   \end{equation}
We have
   \begin{equation}
     \label{eq:20}
      \theta_{pr}(\varphi';n)  =\theta_{pr}^{(o)}(\varphi';n)  +\theta_{pr}^{(e)}(\varphi';n),
    \end{equation}
    with
    \begin{equation}
      \label{eq:24}
      \begin{split}
        &\theta_{pr}^{(o)}(\varphi';n) 
      \\
      &= -\Big(\frac{n}{2}\Big)^{1/2} \sum_{j,j'=1}^{n}
      \sin(\pi k_j) \Big[ \widehat{(\varphi')_{n,o}}(j-j')
      -\widehat{(\varphi')_{n,o}}(j+j')\Big]
      \frac{\la_{j'}  (\la_j - \la_{j'}) }{\theta (\la_j,\la_{j'})}  F_{j,j'},
    \end{split}
  \end{equation}
    and
    \begin{equation}
      \label{eq:25}
      \begin{split}
        & \theta_{pr}^{(e)}(\varphi';n) \\
        &
        =
        -\Big(\frac{n}{2^3}\Big)^{1/2}\sum_{j,j'=1}^{n} 
        \Big[ \widehat{(\varphi')_{n,e}}(j+j') -\widehat{(\varphi')_{n,e}}(j-j')  \Big]
         \frac{ \la_j\la_{j'}  (\la_j - \la_{j'}) }{\theta (\la_j,\la_{j'})} F_{j,j'}.
      \end{split}
    \end{equation}

   By a symmetry argument,   taking into account parity $F_{-j,j'} =F_{j,-j'}=-F_{j,j'}$,
   interchanging
the roles of indices $j$, $j'$ in \eqref{eq:25}, we conclude that  $
\theta_{pr}^{(e)}(\varphi';n) =0$, Consequently 
 %  \begin{equation}
%       \label{eq:27}
%     \theta_{pr}(\varphi';n)  =\theta_{pr}^{(o)}(\varphi';n)  .
%     \end{equation}
% We have
    \begin{equation}
      \label{eq:28}
      \begin{split}
         \theta_{pr}(\varphi';n)  = \theta_{pr}^{(o)}(\varphi';n) 
      =& - \gamma \Big(\frac{n}{2^3 }\Big)^{1/2} \sum_{y=1}^n
      \lang \left(\nabla^* p_y\right)^2\rang_t
    \\
      \times \sum_{j,j'=-n-1}^{n}\sin(\pi k_j)\frac{\gamma_{j'}^2  (\gamma_j^2 - \gamma_{j'}^2) }
      {\theta (\gamma_j^2,\gamma_{j'}^2)} & \widehat{(\varphi')_{n,o}}(j-j')\frac{\cos(\pi u_y (j-j')) - \cos(\pi u_y(j+j'))}{n+1}
      \\
      =: &  \theta_{pr,-}(\varphi';n) -  \theta_{pr,+}(\varphi';n)
      \end{split}
    \end{equation}
    It turns out, see \cite[(7.20)]{KO-supp}, that only $\theta_{pr,-}(\varphi';n)$ contributes, i.e.
    \begin{equation}
   \label{eq:42}
    \lim_{n\to\infty}   \theta_{pr,+}(\varphi';n) = 0.
  \end{equation}
Using elementary  trigonometric formulas we conclude that
    \begin{equation}
      \label{eq:34}
      \begin{split}
         \theta_{pr,-}(\varphi';n) 
    =   - \frac{\gamma n^{1/2}}{2^{3/2} (n+1)} 
      & \sum_{\ell,j'=-n-1 }^{n } 
      \frac{\sin(\pi k_{\ell + j'}) \sin^2\left(\frac{\pi k_{j'}}{2}\right)
            \sin\left(\frac{\pi k_{\ell +2j'}}{2}\right) \sin\left(\frac{\pi k_\ell}{2}\right)}
       {\sin^2\left(\frac{\pi k_{\ell +2j'}}{2}\right)
         \sin^2\left(\frac{\pi k_\ell}{2}\right) + 2^3 \gamma^2\Gamma(k_{\ell+j'},k_{j'})}
    \\
      &\times  \sum_{y=1}^n
      \lang \left(\nabla^* p_y\right)^2\rang_t \widehat{(\varphi')_{n,o}}(\ell) \cos(\pi u_y \ell),
      \end{split}
    \end{equation}
    where
   \begin{equation}
      \label{Gjj} \Gamma(k_j,k_{j'}) = \sin^2\left(\frac{\pi k_j}{2}\right) \sin^2\left(\frac{\pi k_{j'}}{2}\right) 
           \left(\sin^2\left(\frac{\pi k_j}{2}\right)
             +\sin^2\left(\frac{\pi k_{j'}}{2}\right) \right).
         \end{equation}

 % Suppose that $\varphi\in C^\infty_c(0,1)$.
 By \cite[Lemma B.1]{KLO23}, for any $k>0$ we have for some constant $C>0$:
\begin{equation}
  \label{eq:19}
  |\widehat{(\varphi')_{n,o}} (j)| \le  \frac{C}{\chi_n^k(j)},\quad j\in\bbZ,\,n=1,2,\ldots, \qquad \iota = o, e,
\end{equation}
where $\chi_n$ is $2n+2$-periodic extension of the function
$$
\chi_n(j)=(1+j)\wedge (2n+2-j), \quad j=0,\ldots,2n+1.
$$
In addition, if $\kappa\in(0,1)$, then there exists $C>0$ such that
\begin{equation}
  \label{eq:18a}
 \sup_{|j|\le n^{\kappa}}|\widehat{(\varphi')_{n,o}} (j) - \hat \varphi'_s(j)| \le \frac{C}{n^{1-\kappa}}.
\end{equation}}
{where $\hat \varphi'_s(j) = \sqrt 2 \int_0^1 \varphi(u) \sin(\pi j u) \dd u$.}

         Choose $\kappa\in(0,1)$. We further adjust the parameter later on.
    Thanks to \eqref{eq:19} we can consider only the terms
    corresponding to $|\ell| \le
    n^{\kappa}$.     
 Since $|k_\ell|=\frac{|\ell|}{n+1}\le (n+1)^{\kappa-1}$ we   use approximate
  equalities
  \begin{align}
    \label{sinl}
&\sin\Big(\frac{\pi k_\ell}{2}\Big)\approx
\frac{\pi \ell}{2(n+1)}\quad \mbox{and}\quad 2k_{j'}+k_{\ell}
                   \approx 2k_{j'}.
  \end{align}
  Then,
  \begin{equation}
    \label{eq:65}
    \begin{split}
      \Gamma(k_{j'+\ell},k_{j'}) 
          \approx  2 \sin^6\left(\frac{\pi k_{j'}}{2}\right) 
        \end{split}
  \end{equation}
  and, as a result, obtain
  \begin{equation}
    \label{eq:36}
    \begin{split}
          \frac{1}{n}\theta_{pr,-}(\varphi';n)  
                &=  - \frac{\gamma  }{2^2 (n+1)^{3/2}} 
                               \sum_{|\ell|\le n^\kappa}  \pi
                               \ell\widehat{(\varphi')_{n,o}}(\ell) \\
                               &
                               \times \sum_{j'=-n-1}^{n}
      \frac{\sin^2\left(\pi k_{j'}\right)}
       {\cos^2\left(\frac{\pi k_{j'}}{2}\right)
         \left(\frac{\pi\ell}{n}\right)^2 + 2^4\gamma^2 \sin^4\left(\frac{\pi k_{j'}}{2}\right)}
    \\
      &\times   \frac{1 }{ 2(n+1) } \sum_{y=1}^n
      \lang \left(\nabla^* p_y\right)^2\rang_t 2^{1/2}\cos(\pi u_y \ell)
      + o_n(1)\\
       \end{split}
     \end{equation}
  Choose   $\delta\in(0,1)$. Observe that if $|j'|\ge \delta n$ the denominator in the
  last expression is larger than $c\gamma^2 {\delta^4}$ for some $c>0$.
  Because of the factor $n^{-3/2}$ in front and the energy bound
  \eqref{eq:energyb},  the respective  expression vanishes, as
  $n\to+\infty$. As a result,  we can write \eqref{eq:36} as
 \begin{equation}
   \label{eq:37}
   \begin{split}
      \frac{1}{n}\theta_{pr,-}(\varphi';n) &=-\frac{\gamma}{2^2 (n+1)^{3/2}}
        \sum_{|\ell|\le n^\kappa} \pi
        \ell\widehat{(\varphi')_{n,o}}(\ell)\\
        &
        \times \sum_{|j'|\le \delta n }
      \frac{\sin^2\left(\pi k_{j'}\right)}
       { \left(\frac{\pi\ell}{n}\right)^2 \cos^2\left(\frac{\pi k_{j'}}{2}\right)
        + 2^4 \gamma^2 \sin^4\left(\frac{\pi k_{j'}}{2}\right)}
    \\
      &\times \frac{1}{2( n+1)} \sum_{y=1}^n
      \lang \left(\nabla^* p_y\right)^2\rang_t  \sqrt 2
      \cos(\pi u_y \ell) + o_n(1).
    \end{split}
  \end{equation}
  Using the  approximate equalities $\sin\Big(\frac{\pi k_{j'}}{2}\Big)\approx
\frac{\pi k_{j'}}{2}$ and $\cos\Big(\frac{\pi k_{j'}}{2}\Big)\approx
1$, valid for a sufficiently small $\delta$, we can rewrite it as
  \begin{align*}
   &
        \frac{1}{n}\theta_{pr,-}(\varphi';n) = - \frac{\gamma}{2^2 n^{3/2}}
       \sum_{|\ell|\le n^\kappa}  \pi \ell\widehat{(\varphi')_{n,o}}(\ell)\sum_{j'=-\delta n}^{\delta n}
      \frac{\left(\pi k_{j'}\right)^2}{ \left(\frac{\pi\ell}{n}\right)^2 +  \gamma^2 (\pi k_{j'})^4}
    \\
      &\times \frac{1}{2(n+1)} \sum_{y=1}^n
      \lang \left(\nabla^* p_y\right)^2\rang_t  \sqrt 2
        \cos(\pi u_y \ell) + o_n(1)\\
    &
        = - \frac{\gamma}{2^2 n^{3/2}}
       \sum_{|\ell|\le n^\kappa}  \pi \ell\widehat{(\varphi')_{n,o}}(\ell)\sum_{j'=-\delta n}^{\delta n}
      \frac{\left(\frac{\pi j'}{n+1}\right)^2}{ \left(\frac{\pi\ell}{n}\right)^2 +  \gamma^2 \left(\frac{\pi j'}{n+1}\right)^4}
    \\
      &\times \frac{1}{2(n+1)} \sum_{y=1}^n
      \lang \left(\nabla^* p_y\right)^2\rang_t  \sqrt 2
        \cos(\pi u_y \ell) + o_n(1)
       \end{align*}

       \begin{align*}
          &
        = - \frac{\gamma}{2^2 }
       \sum_{|\ell|\le n^\kappa}  \pi \ell\widehat{(\varphi')_{n,o}}(\ell) \Bigg[ \frac{1}{ n^{1/2}}\sum_{j'=-\delta n}^{\delta n}
      \frac{\left(\frac{\pi j'}{(n+1)^{1/2}}\right)^2}{ \left(\pi\ell\right)^2 +  \gamma^2 \left(\frac{\pi j'}{(n+1)^{1/2}}\right)^4}\Bigg]
    \\
      &\times \frac{1}{2(n+1)} \sum_{y=1}^n
      \lang \left(\nabla^* p_y\right)^2\rang_t  \sqrt 2
        \cos(\pi u_y \ell) + o_n(1)\\
    =   -\frac{\gamma}{2^2 }
      & \sum_{|\ell|\le n^\kappa} \pi \ell\widehat{(\varphi')_{n,o}}(\ell)\int_{\bbR}
      \frac{(\pi u)^2 \dd u}{ \left(\pi\ell\right)^2 +  \gamma^2 (\pi u)^4}
    \\
      &\times \frac{1}{2( n+1)} \sum_{y=1}^n
      \lang \left(\nabla^* p_y\right)^2\rang_t  \sqrt 2
      \cos(\pi u_y \ell) + o_n(1).
          \end{align*}
{Using parity of $\pi \ell\widehat{(\varphi')_{n,o}}(\ell)$} and  changing variables $v=\left(\frac{\pi \ga }{\ell}\right)^{1/2}u$ we
conclude that
 \begin{equation}
    \label{eq:38a}
    \begin{split}
      \frac{1}{n}\theta_{pr,-}(\varphi';n) &
      =  - {\frac{1}{2 \pi\gamma^{1/2}}}
      \sum_{0<\ell\le n^\kappa} (\pi \ell)^{1/2}\widehat{(\varphi')_{n,o}}(\ell)
            \int_{\mathbb R}  \frac{v^2 \dd v}{ 1 + v^4}
    \\
      &\times \frac{1 }{2(n+1)} \sum_{y=1}^n
      \lang \left(\nabla^* p_y\right)^2\rang_t  \sqrt 2
      \cos(\pi u_y \ell) + o_n(1).
       \end{split}
   \end{equation}

  %     \\
  %     =   \frac{\gamma^{1/2}}{2 n^{1/2}\pi}
  %     & \sum_{|\ell|\le n^\kappa}
  %     \int_{\frac{(1-\ell)\pi\gamma^{1/2}}{n}\vee (-\delta\pi\gamma^{1/2})}^{\delta\pi\gamma^{1/2}}
  %     \frac{\gamma^{-1} v^2 \dd v}{ \left(\frac{\pi\ell}{n}\right)^2 + v^4}
  %   \\
  %     &\times \frac{\pi \ell}{2n} \sum_{y=1}^n
  %     \lang \left(\nabla^* p_y\right)^2\rang_t \widehat{(\varphi')_{n,o}}(\ell)\sqrt 2
  %     \cos(\pi u_y \ell) + o_n(1)\\
  %     =   \frac{1}{2 n^{1/2}\pi \gamma^{1/2}}
  %     & \sum_{|\ell|\le n^\kappa} \left(\frac{n}{\pi |\ell|} \right)^{1/2}
  %     \int_{(\frac{(1-\ell)}{n}\vee (-\delta)) c_\ell \sqrt n}^{\delta c_\ell \sqrt n}
  %     \frac{v^2 \dd v}{ 1 + v^4}
  %   \\
  %     &\times \frac{\pi \ell}{2n} \sum_{y=1}^n
  %     \lang \left(\nabla^* p_y\right)^2\rang_t \widehat{(\varphi')_{n,o}}(\ell)\sqrt 2
  %     \cos(\pi u_y \ell) + o_n(1)\\
  %    =   \frac{1}{2 \pi\gamma^{1/2}}
  %     & \sum_{|\ell|\le n^\kappa} 
  %     \int_{(\frac{(1-\ell)}{n}\vee (-\delta)) c_\ell \sqrt n}^{\delta c_\ell \sqrt n}
  %     \frac{v^2 \dd v}{ 1 + v^4}
  %   \\
  %     &\times \frac{1 }{2n} \sum_{y=1}^n
  %     \lang \left(\nabla^* p_y\right)^2\rang_t \rm{sign}(\ell) (\pi|\ell|)^{1/2}\widehat{(\varphi')_{n,o}}(\ell)\sqrt 2
  %     \cos(\pi u_y \ell) + o_n(1).
  %   \end{split}
  % \end{equation}

  Using the residue theorem one can calculate
  \begin{equation}
    \label{eq:39}
       \frac 1{2\pi}  \int_{\mathbb R} \frac{v^2 \dd v}{ 1 + v^4} = \frac 1{2^{3/2}}.
  \end{equation}
To finish the identification of the evolution of the bulk terms we
need the following {energy equipartition} result, see \cite[Proposition 7.2]{KO-supp}.
  \begin{proposition}
    For any $\varphi \in C[0,1]$, that is compactly supported in
    $(0,1)$, we have 
    \label{prop:energy}
    \begin{equation}
   \lim_{n\to\infty} \frac{1}{n} \sum_{y=1}^n
    \int_0^t \varphi_y \mathbb E_n
    \left(\left(\nabla^* p_y(s)\right)^2 - 2 \mathcal E_y(s)\right)  \; \dd s = 0.
    \label{eq:40}
  \end{equation}
\end{proposition}

Then from \eqref{eq:18a},  \eqref{eq:39} and  \eqref{eq:40} we
conclude that 
  \begin{equation}
    \label{eq:32}
    \begin{split}
      &  \frac{1}{n}\theta_{pr,-}(\varphi';n) \\
    \\
    &=
   - \frac{1}{(2^3 \gamma)^{1/2}} \int_0^t \dd s
    \frac{1}{n} \sum_{y=1}^n  \mathbb E_n\left(\mathcal E_y(s)\right)
    \sum_{\ell=1}^{+\infty} c_\ell (u_y) (\pi\ell)^{1/2}\widehat{(\varphi')}_s(\ell) +o_n(1).
  \end{split}
\end{equation}
 
Since $\widehat{(\varphi')}_s(\ell)
=-\pi\ell\hat{\varphi}_c(\ell)  $ this ends the proof of Proposition \ref{lem-pr}.
We can also rewrite \eqref{eq:32} in the form
\begin{align*}
    &  \frac{1}{n}\theta_{pr,-}(\varphi';n)  
     = \frac{1}{(2^3 \gamma)^{1/2}} \sum_{\ell=1 }^{+\infty} \int_0^t \dd s\;
       \hat{\cal E}(s,\ell)
        (\pi|\ell|)^{3/2} \widehat{\varphi}_{c}(\ell)  +o_n(1),
\end{align*}
 with
  $ \widehat{\cal E}(s,\ell):=\frac{1}{n+1}\sum_{y=0}^n 
    \bbE_n[{\cal E}_y(s)]  c_\ell(u_y)$ and the conclusion of the Proposition
                   \ref{lem-pr} follows.
\qed

\subsection{{Asymptotics for the  boundary layer terms
  $\pi^{(pr)}_{\iota}(\varphi'; n)$}}

In the present section we characterize the behavior of the terms
$\frac1n\pi^{(pr)}_{\iota} (\varphi';n)$ appearing in
\eqref{eq:energy}. They describe the interactions between the bulk and
boundary points $v=0,1$.

 Consider $ \tilde b^{(pr)}_{z,j}(s) $,  $ \tilde b^{(p)}_{z,j}(s) $
 - the Fourier coefficients of the fields corresponding to the
 boundary bulk interactions - see \eqref{bpr} and \eqref{bpp}.
Thanks to the   $L^2$ estimates of the covariances, see \eqref{eq:9},
we obtain
\begin{align}
  \label{011512-25}
                                              & (n+1)^{1/2}\sum_{j=0}^n \Big(\tilde b_{z,j}^{(p)}(t) \Big)^2 =  (n+1)^{1/2}\sum_{x=0}^n \Big({
                 b}^{(p)}_{z,x}( t)\Big)^2\le C,
\end{align}
for $n=1,2,\ldots,$ and $z=0,n$.
For $v=0,1$ we define the functions
$ \frak{ b}^{(p,v)}_{n } :[0,+\infty)^2\to\bbR 
 $
 given by
 \begin{align}
   \label{021612-25}
   &{\frak b}^{(p,v)}_{n }  (t, \varrho)=0, \quad \varrho\ge (n+1)^{1/6}\pi ,\\                
      &  {\frak b}^{(p,v)}_{n } (t, \varrho)=
       (n+1)^{1/2}\tilde b_{nv,j}^{(p)}(t), 
      \,0\le j\le  (n+1)^{2/3}, \,  \varrho\in
  \big[\varrho_j, \varrho_{j+1}\big), \quad \mbox{where }\notag\\
              &\varrho_j:= \frac{j\pi}{(n+1)^{1/2}} \notag.
  \end{align}
From \eqref{011512-25} we immediately conclude  that for any $t>0$
there exists $C>0$ such that
\begin{align}
  \label{011612-25}
      \int_0^t\dd s\int_{0}^{+\infty}\big[{\frak
    b}^{(p,v)}_{n}(s, \varrho )\big]^2\dd \varrho\le C, \quad \,v=0,1,
    \, n=1,2,\ldots.
\end{align}
{We have that $\pi^{(pr)}_{r} (\varphi';n)=0$ and by \cite[Proposition 5.2]{KO-supp}
it turns out that $\frac{1}{n}\big(\pi^{(pr)}_{pr}+ \pi^{(pr)}_{rp} (\varphi';n)\big)$ is negligible.
Therefore, the only term that needs to be considered  is $ \frac{1}{n}\pi^{(pr)}_{p} (\varphi';n) ( n)$. }

%The following result holds, see \cite[Proposition 5.2]{KO-supp}
\begin{proposition}
\label{prop011512-25}
%  We have $\pi^{(pr)}_{r} (\varphi';n)=0$ . In addtion, 
  As $n\to+\infty$, the following asymptotics hold
      \begin{align*}
       %     & \frac{1}{n}\big(\pi^{(pr)}_{pr} (\varphi';n)+\pi^{(pr)}_{rp} (\varphi';n)\big)=o_n(1),\\
            %
                  &
                \frac{1}{n}\pi^{(pr)}_{p} (\varphi';n) ( n)  
                  =   \frac{2^{1/2}\tilde \ga   }{ \pi}    
                    \sum_{v=0,1}\sum_{\ell=1}^{+\infty}c_\ell(nv)(\pi \ell)^2\hat{\varphi
                                                                                                     }_c(\ell) \\
        &
        \qquad \qquad \qquad    \times\int_0^t\dd s\int_0^{+\infty}
                  \frac{{\frak b}^{(p,v)}_{n}(s, \varrho) \dd\varrho 
        }{
         ( \pi\ell )^2  
    +  \ga^2 
           \varrho ^4  }  +o_n(1).
      \end{align*}
    \end{proposition}

  { From \eqref{eq:energy}, Propositions \ref{lem-pr} and
   \ref{prop011512-25} it follows that for any $\varphi\in C_c^\infty(0,1)$
   \begin{align}
     \label{051612-25}
    &
    {\cal E}_n [t; \varphi] -  {\cal E}_n [0; \varphi]      
     =
   -\frac{1}{(2^3 \gamma)^{1/2}} \sum_{\ell=1}^{+\infty} (\pi \ell)^{3/2}\hat{\varphi     }_c(\ell)\int_0^t \widehat{\cal E}(s,\ell)\dd s
     \\
      &
             +  \frac{2^{1/2}\tilde \ga   }{ \pi}    
                    \sum_{v=0,1}\sum_{\ell=1}^{+\infty}c_\ell(nv)(\pi \ell)^2\hat{\varphi     }_c(\ell)
        \int_0^t\dd s\int_0^{+\infty}
                  \frac{{\frak b}^{(p,v)}_{n}(s, \varrho) \dd\varrho 
        }{
         ( \pi\ell )^2  
    +  \ga^2 
           \varrho ^4  }  
      +o_n(1)\notag
    \end{align}
This is not a closed equation. }
    We need to find an expression for ${\frak b}^{(p,v)}_{n}(s, \varrho)$
    in terms of the energy {field} in the bulk.

The asymptotics ${\frak b}^{(p,v)}_{n}(s, \varrho)$, as $n\to+\infty$, involves
          the covariances between momenta at the boundary and
          stretches in the bulk. To describe it we  define   sequences of functions
$$
{\frak b}^{(pr,v)}_{n } :[0,+\infty)^2\to\bbR,\quad  v=0,1 
$$
as follows.       For $\varrho\ge (n+1)^{2/3}
\pi$ and $t\ge0$ we let                      ${\frak b}^{(pr,v)}_{n,z}(t, \varrho)=0$.
 For $ 0\le j\le  (n+1)^{2/3}$, $t\ge0,\, \varrho\in
  \big[\varrho_j, \varrho_{j+1}\big)$ we  let
                                             (see  \eqref{bpr})
\begin{align}
    \label{tbr}  {\frak b}^{(pr,v)}_{n }(t, \varrho)=
     (n+1)^{1/2}\tilde b_{nv,j}^{(pr)}(t),\quad   v=0,1.
\end{align}
By the Plancherel identity  and \eqref{eq:9} (see \eqref{011612-25})
for any $t>0$
there exists $C>0$ such that   
\begin{align*}
   & 
\int_0^t\dd s\int_{0}^{+\infty}\big[{\frak
  b}^{(\iota,v)}_{n }(s, \varrho)\big]^2\dd
     \varrho\le    (n+1)^{1/2}
     \sum_{x=0}^n \int_0^t\big( b_{nv,x}^{(\iota)}(s)\big)^2\dd s\le C
\end{align*}
for $n=1,2,\ldots$. The following result has been established in
\cite[Proposition 9.1]{KO-supp}
        \begin{proposition}
    \label{prop010901-25}
    For any  test function
       $f\in    L^2[0,+\infty)$, $t>0$ and $v=0,1$
       we have
       \begin{align}
         \label{031612-25}
        &(1+2\tilde\ga)  \int_0^t\dd s\int_0^{+\infty}{\frak b}^{(p,v)}_{n}(s,\varrho)f(\varrho)\dd \varrho   
  \\
&     = \sqrt{2} \int_0^t\dd s\int_{0}^{+\infty}\Bigg( T_v- \sum_{\ell=0}^{+\infty} 
   \frac{\ga^2  \varrho^4 c_\ell(v) \widehat{\cal E}(s,\ell)}{ (\ell\pi)^2 
    +\ga^2 \varrho^4    } \Bigg)f(\varrho)\dd \varrho \notag\\
   &
     +\frac{\tilde\ga}{\pi} \int_0^t\dd s \int_{0}^{+\infty}
  {\frak b}^{(pr,v)}_{n}( s,\varrho){\frak T}f(\varrho)\dd \varrho +o_n(1) .\notag
   \end{align}
 \end{proposition}

The operator appearing on the right hand side of \eqref{031612-25} is
defined  for $f\in C^1_c[0,+\infty) $ - the space of compactly supported $C^1$-smooth
functions - using the formula
 \begin{equation*}
            {\frak T}f(\varrho)=2\int_0^{+\infty}
      \frac{[f(\varrho')-f(\varrho)]\varrho}{(\varrho-\varrho')(\varrho+\varrho')}\dd
      \varrho',\quad f\in C^1_c[0,+\infty).
    \end{equation*}
It    extends continuously to a linear operator   ${\frak
  T}:L^p[0,+\infty)\to L^p[0,+\infty)$ for any $p\in(1,+\infty)$, see
\cite[Theorem C.1]{KO-supp}. In addition,  its adjoint {${\frak
  T}^\star:L^q[0,+\infty)\to L^q[0,+\infty)$, with $1/p+1/q=1$,}  is
the continuous extension of
\begin{equation}
  \label{Ts}
{\frak T}^\star f(\varrho):=2\int_0^{+\infty}\frac{   [ \varrho' f(\varrho')-\varrho f(\varrho)]  \dd
  \varrho' }{(\varrho' -\varrho)(\varrho' +\varrho)},\quad f\in C_c^1[0,+\infty).
\end{equation} and {for $p=q=2$ we have}
\begin{equation}
  \label{eq:id}
 {\frak T}^\star {\frak T}=\pi^2I,
\end{equation}
where $I$ is the identity operator on  $L^2[0,+\infty) $.

 To describe the asymptotics of  ${\frak  b}^{(pr,v)}_{n}( s,\varrho)$
  the following result has been established in \cite[Proposition 9.2]{KO-supp}
   \begin{proposition}
    \label{prop011001-25}
    For   any test function
       $f\in    L^2[0,+\infty) $, $t>0$ and $v=0,1$
       we have
\begin{align*}
  & \int_0^t\dd s\int_0^{+\infty}{\frak  b}^{(pr,v)}_{n}( s,\varrho)f(\varrho)\dd \varrho
  \\
  &
    =-\frac{\tilde\ga}{\pi} \int_0^t\dd s\int_{0}^{+\infty} 
  {\frak b}^{(p,v)}_{n}( s,\varrho) {\frak T}^\star f(\varrho)\dd \varrho + o_n(1) . 
\end{align*}
\end{proposition}
 Finally from Propositions \ref{prop010901-25} and \ref{prop011001-25}
 we conclude the folowing.
\begin{theorem}
 \label{thm011001-25}
   For   any  test function
       $f\in    L^2[0,+\infty)$, $t>0$
 and $v=0,1$      we have
\begin{align*}
     & \int_0^t\dd s \int_0^{+\infty}{\frak b}^{(p,v)}_{n}(s,\varrho)f(\varrho)\dd \varrho   
 = \frac{\sqrt{2}}{(1+\tilde\ga)^2 } \\
   &
   \qquad    \times \int_0^t\dd s\int_{0}^{+\infty}\Bigg(T_v-
                      \sum_{\ell=0}^{+\infty} 
     \frac{\ga^2  \varrho^4 c_\ell(v) \widehat{\cal E}_n(s,\ell)}{ (\ell\pi)^2 +\ga^2  \varrho^4    }
     \Bigg) f(\varrho) \dd \varrho + o_n(1)  \notag
\end{align*}
where $\widehat{\cal E}_n(t,\ell)$  are the Fourier cosine coeffcients of
  the energy functional.
\end{theorem}
 {Combining \eqref{051612-25} and   Theorem
   \ref{thm011001-25}  
   concludes  the proof of \eqref{eq:conv-temp}.}

 % {\sout{ Equipartition property, claimed in \eqref{eq:conv-temp2}, is shown
 %  in \cite[Section 11.2]{KO-supp}. }}

\subsection*{Funding declaration} T. Komorowski acknowledges the support of
the NCN grant 2024/53/B/ST1/00286.

 \subsection*{Data availability statement}
We do not analyse or generate any datasets, because our work proceeds within a theoretical and mathematical approach.

 \subsection*{Conflict of interest statement.} The authors declare no conflicts of interest
regarding this manuscript.


\begin{thebibliography}{99}

% \bibitem{aron} N. Aronszajn; K. T. Smith (1961). {\it Theory of Bessel potentials I}. Ann. Inst. Fourier. 11. 385-475. doi:10.5802/aif.116

\bibitem{bborev} G. Basile, C. Bernardin, S. Olla, {\it A momentum
  conserving model with anomalous thermal conductivity in low
  dimension}, Phys. Rev. Lett. \textbf{96}, 204303 (2006),
DOI 10.1103/Phys-RevLett.96.204303.


\bibitem{bbo2}  G. Basile, C. Bernardin, S. Olla, {\it Thermal Conductivity
  for a Momentum Conservative Model}, Comm. Math. Phys. \textbf{287},
  67--98, (2009).

\bibitem{babo} G. Basile, A. Bovier,
  {\it Convergence of a kinetic equation to a fractional diffusion equation,}
 Markov Proc. Rel. Fields \textbf{16}, 15-44 (2010);

  
\bibitem{BOS} G. Basile, S. Olla, H. Spohn, {\it{Energy transport in
    stochastically perturbed lattice dynamics}}, Arch. Rat. Mech. Anal,
  Vol. 195, no. 1, 171-203, 2009. 

   \bibitem{BBJKO} G. Basile, C. Bernardin, M. Jara, T. Komorowski, S. Olla,
    {\it{ Thermal conductivity in harmonic lattices with random
      collisions}}, in ``Thermal transport in low dimensions: from
    statistical physics to nanoscale heat transfer'', S. Lepri ed.,
   \textbf{Lecture Notes in Physics} \textbf{921}, chapter 5,
   Springer 2016.\\
   https://doi.org/10.10007/978-3-319-29261-8-5

  \bibitem{BGJ16}    Bernardin, C., Gonçalves, P., Jara, M.,
    {\it 3/4-Fractional Superdiffusion in a System of Harmonic Oscillators Perturbed by a Conservative Noise}.
    Arch Rat. Mech. Anal 220, 505–542 (2016).
    https://doi.org/10.1007/s00205-015-0936-0

    
  \bibitem{bo11} C. Bernardin, S. Olla,
    {\it Transport Properties of a Chain of Anharmonic Oscillators with random flip of velocities},
    J. Stat. Phys. 145, 1224-1255, 2011.
    https://doi.org/10.1007/s10955-011-0385-6


 \bibitem{BCGS-24} C. Bernardin, P Cardoso, P Gon\c{c}alves, S. Scotta,
   {{\it Hydrodynamic limit for a boundary driven
       super-diffusive symmetric exclusion}},
   Stoch. Proc. and their Appl., Vol. 165, Pages 43-95,   2023.

 \bibitem{BGJO-21}  C. Bernardin, P. Gonçalves, and B. Jiménez-Oviedo.
   {\it A microscopic model for a one parameter class of fractional
      Laplacians with Dirichlet boundary conditions}.
    Arch. Ration. Mech. Anal., 239(1):1–48, 2021.

% \bibitem{DS} N. Dunford, J. T. Schwartz {\it Linear Operators,
%     II. Spectral Theory}, Interscience Publishers, 1963

%  \bibitem{GR} I.S. Gradshteyn and I.M. Ryzhik, {\it Table of Integrals, Series, and Products},
%    Seventh Edition, Academic Press-Elsevier 2007.

\bibitem{BLR} F. Bonetto, J.L. Lebowitz, and L. Rey-Bellet,
  {\it Fourier's Law: a Challenge for Theorists},
  Mathematical Physics 2000,
  Edited by A. Fokas, A. Grigoryan, T. Kibble and B. Zegarlinsky,
  Imperial College Press, 128-151 (2000).

 \bibitem{kjo} M. Jara, T. Komorowski,  S. Olla, {\it{A limit theorem for an
   additive functionals of Markov chains}}, Annals of Applied Probability
   \textbf{19}, No. 6, 2270-2300, 2009.

   
 \bibitem{JKO15}   Jara, M., T. Komorowski, S. Olla,
   {{\it  Superdiffusion of energy in a chain of harmonic
  oscillators with noise.}} Commun. Math. Phys. {\bf 339}, 407-453
(2015), https://doi.org/10.1007/s00220-015-2417-6
   
   % \bibitem{kelley} Kelley, J. L. (1991), {\it General topology},
   % Springer-Verlag, ISBN 978-0-387-90125-1.

 \bibitem{kl}  Kipnis, C., Landim, C., {\it  Scaling Limits of
     Interacting Particle Systems},
   Grundlehren der mathematischen Wissenschaften 320, Springer,
   https://doi.org/10.1007/978-3-662-03752-2


 \bibitem{KO-supp} T. Komorowski, S. Olla,
   {\it Thermal boundary conditions in fractional superdiffusion of energy.}
   (2025). Available at https://arxiv.org/pdf/2505.06952v5.


 \bibitem{KLO23} T.  Komorowski, J. L. Lebowitz, S. Olla, {\it Heat flow in
   a periodically forced, thermostatted chain}, Comm. Math. Phys.,  400, pp 2181-2225 (2023)
 https://doi.org/10.1007/s00220-023-04654-4.

  % \bibitem{KLO23-2} T. Komorowski, J.L. Lebowitz, S. Olla,
  % {\it Heat flow in a periodically forced, thermostatted chain II},
  %  {J. Stat. Phys.}, 2023, 190 (4), pp.87.
  % https://doi.org/10.1007/s10955-023-03103-9.

\bibitem{KO20}   T. Komorowski, S. Olla,
  {\it Kinetic limit  for a chain of harmonic oscillators with a point Langevin thermostat.}
   Journ. of Funct. Analysis,  {\bf 279} (2020), Article \# 108764,
   https://doi.org/10.1016/j.jfa.2020.108764

 \bibitem{KOrev}   T. Komorowski, S. Olla,
   {\it Thermal Boundaries in Kinetic and Hydrodynamic Limits},
   in ”Recent advances in kinetic equations and applications”,
  F. Salvarani ed., Springer
  INdAM Series 48, 2021, pp 253–288,
  %ISSN 2281-518X, ISBN 978-3-030-82945-2,
https://doi.org/10.1007/978-3-030-82946-9 11

  % \bibitem{KO20a}   T. Komorowski, S. Olla, {\it Asymptotic Scattering by Poissonian Thermostats}
  %   Annales Henri Poincar\'e {\bf 23} (2022), 3753-3790, https://doi.org/10.1007/s00023-022-01173-1.

  \bibitem{KOR20}  T. Komorowski, S. Olla, L. Ryzhik,
    {\it Fractional diffusion limit for a kinetic equation with an interface.},
    Annals of Probability 2020, Vol. {\bf 48}, No. 5, 2290-2322. https://doi.org/10.1214/20-AOP1423

  \bibitem{KORS20} T. Komorowski, S. Olla, L. Ryzhik, H. Spohn,
    {\it High frequency limit  for a chain of harmonic oscillators with a point Langevin thermostat}.  
Archive for Rational Mechanics and Analysis volume {\bf 237}, pages
497-543 (2020), https://doi.org/10.1007/s00205-020-01513-7

\bibitem{KB18}   {Kundu, A. and Bernardin, C. and Saito, K. and
              Kundu, A. and Dhar, A.},
     {\it Fractional equation description of an open anomalous heat conduction set-up},
   {J. Stat. Mech. Theory Exp.}, {2019},  {013205, 28},
       {https://doi.org/10.1088/1742-5468/aaf630},


\bibitem{KW17}          Kwa\'snicki, M.,
     {\it Ten equivalent definitions of the fractional {L}aplace
              operator},
    {Fract. Calc. Appl. Anal.},  {\bf 20},  {2017}, {7--51}.
      

\bibitem{sll} S. Lepri, R. Livi, A. Politi,  {Thermal Conduction in
 classical low-dimensional lattices}, Phys. Rep. \textbf{377}, 1-80 (2003).

\bibitem{llp97} S. Lepri, R. Livi, A. Politi, {Heat conduction in
    chains of nonlinear oscillators}, 
  Phys. Rev. Lett. \textbf{78}, 1896 (1997).

       
\bibitem{LM08}   S. Lepri, C. Mejıa-Monasterio and A. Politi,
  {\it A stochastic model of anomalous heat transport:
analytical solution of the steady state.}
  J. Phys. A: Math. Theor. 42 (2009) 025001,  https://doi.org/10.1088/1751-8113/42/2/025001

  \bibitem{LM10}  S. Lepri, C. Mejıa-Monasterio and A. Politi,  {\it Nonequilibrium dynamics of a stochastic model of anomalous heat transport.}
  J. Phys. A: Math. Theor. 43 (2010) 0650002, https://doi.org/10.1088/1751-8113/43/6/065002 

 % \bibitem{stein} E. Stein, {\it Singular integrals and
 %     differentiability properties of functions}, 1970, Pricenton
 %   Univ. Press.
 
\end{thebibliography}
\end{document}